\documentclass[conference]{IEEEtran}
\IEEEoverridecommandlockouts
\usepackage{cite}
\usepackage{amssymb}
\usepackage[figuresright]{rotating}
\usepackage{multirow}
\usepackage{algorithm}
\usepackage{algorithmicx}
\usepackage{algpseudocode}
\usepackage{subfigure}
\usepackage{amsmath}
\usepackage{amsthm}
\usepackage{amsfonts}
\usepackage{bbm}
\usepackage{color}
\usepackage{url}
\usepackage{caption}

\newtheorem{theorem}{Theorem}
\newtheorem{lemma}{Lemma}
\newtheorem{definition}{Definition}
\newtheorem{example}{Example}

\def\BibTeX{{\rm B\kern-.05em{\sc i\kern-.025em b}\kern-.08em
    T\kern-.1667em\lower.7ex\hbox{E}\kern-.125emX}}
\begin{document}

\title{Sketches-based join size estimation under local differential privacy}

\author{
	\IEEEauthorblockN{
		Meifan Zhang,
		~Xin Liu,
		~Lihua Yin\IEEEauthorrefmark{1}\thanks{\IEEEauthorrefmark{1} Lihua Yin is the corresponding author.}}
    \IEEEauthorblockA{Cyberspace Institute of Advanced Technology, Guangzhou University, Guangzhou, 510006, China}
	\IEEEauthorblockA{zhangmf@gzhu.edu.cn, liuxin22@e.gzhu.edu.cn, yinlh@gzhu.edu.cn}
}

\maketitle

\begin{abstract}
  Join size estimation on sensitive data poses a risk of privacy leakage. Local differential privacy (LDP) is a solution to preserve privacy while collecting sensitive data, but
  it introduces significant noise when dealing with sensitive join attributes that have large domains. Employing probabilistic structures such as sketches is a way to handle large domains, but it leads to hash-collision errors. To achieve accurate estimations, it is necessary to reduce both the noise error and hash-collision error.
  To tackle the noise error caused by protecting sensitive join values with large domains, we introduce a novel algorithm called LDPJoinSketch for sketch-based join size estimation under LDP. Additionally, to address the inherent hash-collision errors in sketches under LDP, we propose an enhanced method called LDPJoinSketch+. It utilizes a frequency-aware perturbation mechanism that effectively separates high-frequency and low-frequency items without compromising privacy. The proposed methods satisfy LDP, and the estimation error is bounded.
Experimental results show that our method outperforms existing methods, effectively enhancing the accuracy of join size estimation under LDP.
\end{abstract}


\section{Introduction}\label{sec:introduction}
Join size estimation, also known as inner product estimation, represents a fundamental statistical problem with applications spanning various domains, including query optimization~\cite{DBLP:conf/sigmod/IzenovDRS21}, similarity computation~\cite{DBLP:conf/kdd/WangQZZWLG19}, correlation computation~\cite{DBLP:conf/sigmod/SantosBCMF21}, and dataset discovery~\cite{DBLP:conf/pods/BessaDFMMSZ23}.
Despite numerous research efforts dedicated to this issue, the data collection process for join size estimation poses inherent privacy risks.
Local differential privacy (LDP)~\cite{Kasiviswanathan2008WhatCW} is a solution to privacy-preserving data collection~\cite{Wang2017LocallyDP,DBLP:journals/isci/ZhangLY23, Duchi2016MinimaxOP}, gaining traction in real-world applications at companies including Apple~\cite{2017LearningWP}, Google~\cite{Erlingsson2014RAPPORRA, Fanti2015BuildingAR}, and Microsoft~\cite{Ding2017CollectingTD}. However, most previous works have primarily focused on basic statistics like frequency and mean estimation, offering limited solutions for more intricate statistical tasks such as join aggregations. 
In a specific study~\cite{Xu2020CollectingAA}, the authors explored join operations on two private data sources. However, it is important to highlight that their approach assumed non-private values for the join attributes.

Join size estimation under LDP holds significance in various scenarios.
  (1) Private similarity computation for data valuation and pricing.
  Estimating the similarity of two streams is a basic application of join size estimation. Private similarity computation is essential in data markets to assess the value of private data from different sources for analysis or learning tasks~\cite{DBLP:journals/tifs/ChristensenPP23}.
  (2) Private correlation computation for dataset search and discovery.
  Join size estimation plays a role in dataset search by identifying relevant factors through joinable columns and computing column correlation~\cite{DBLP:conf/pods/BessaDFMMSZ23}. Private join size estimation benefits relevant industries such as hospitals and genetics research companies in assessing the relevance of private data before collaborating.
  (3) Private approximate query processing.
  Approximate Query Processing (AQP) can enhance response times by returning approximate results. Applying AQP to private data presents a promising approach for efficient and private data analysis~\cite{DBLP:journals/pvldb/BaterP0WR20, DBLP:conf/bigdataconf/OckLK23}, given that exact query answers are unattainable under DP or LDP. Join size estimation, as a specific task of AQP, holds potential to be extended to handle more general AQP tasks, including other join aggregations and predicates, under LDP.

Join size estimation on private join attribute values under LDP is a complex task. \textbf{A common LDP setting} has two main kinds of participants including a large number of users on client-side and one untrusted aggregator on server-side. An LDP workflow can be broken down into two steps: Each user locally \textbf{perturbs} its sensitive value by LDP mechanisms and sends the perturbed value to the aggregator. The server-side then \textbf{aggregates} the perturbed values for subsequent specific analysis.
Join size estimation under LDP comes with several challenges.
\textbf{Challenge I, the join attribute values are sensitive and have a large domain.} It is difficult to perturb the private join attribute values from two data sources while preserving the join size.  Directly perturbing the join values according to LDP mechanisms, such as k-RR~\cite{Wang2017LocallyDP} and FLH~\cite{Cormode2021FrequencyEU}, can introduce significant noise. To tackle the challenge of large domain, previous approaches like Apple-HCMS~\cite{2017LearningWP} and RAPPOR~\cite{Erlingsson2014RAPPORRA} have turned to probabilistic data structures, such as sketches and bloom filters, for estimation. However, these structures reduce the domain while introducing hash-collision errors.
\textbf{Challenge II, reducing the hash-collision error of probabilistic structures while preserving privacy.} Without considering the privacy, the main idea to address hash-collision error is to compute the join size of high-frequency items and low-frequency items separately. This is because collisions involving high-frequency items tend to introduce more significant errors. However, separating high-frequency and low-frequency items while preserving LDP is also not simple. This is because the frequency property of each value is also private, making it a complex problem to tackle.


\textbf{Main idea.} To tackle the challenge I, we propose the LDPJoinSketch, whose main idea is to modify the fast-AGMS sketch to an LDP version, which constructs fast-AGMS sketch for join size estimation in the server-side based on the perturbed values collected from each individual users in the client-side. To tackle the challenge II, we propose LDPJoinSketch+ to reduce the hash-collision error by differently encoding the low-frequency and high-frequency items while satisfying LDP.

\begin{figure}
  \centering
  \includegraphics[width=0.48\textwidth]{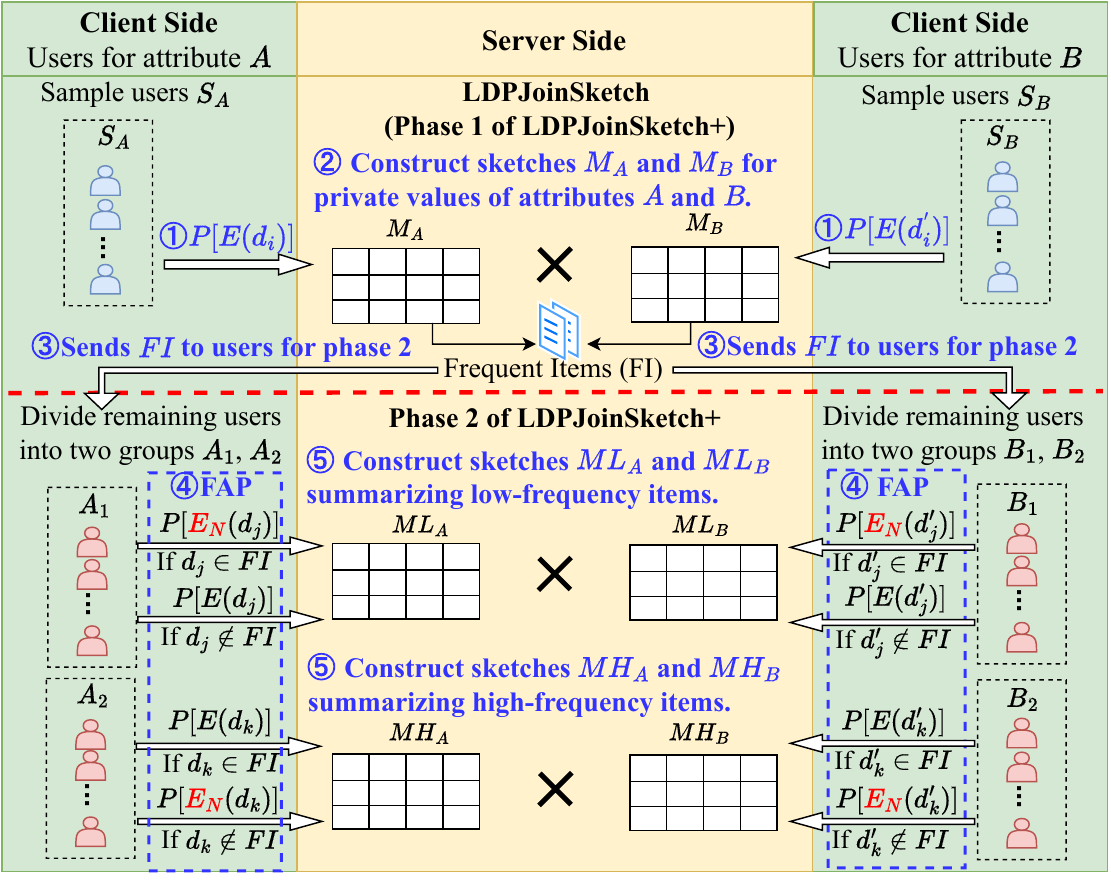}
  \caption{Sketches-based join size estimation under LDP.}
  \label{Fig:intro_framework}
\end{figure}

Fig.~\ref{Fig:intro_framework} shows the workflow of the two proposed methods.
\textbf{The workflow of LDPJoinSketch} adheres to the common LDP pattern. In step $\textcircled{1}$, each user on the client-side encodes and perturbs its sensitive join attribute value $d_i$ and  transmits the processed value $P(E(d_i))$ to the server, here $P(E(d_i))$ denotes the encoding and perturbation of value $d_i$. In step $\textcircled{2}$, the server receives all perturbed values and constructs sketches $M_A$ and $M_B$ for attributes $A$ and $B$, respectively. Ultimately, the join size can be estimated from the product $M_A\times M_B$.
\textbf{The workflow of the two-phase algorithm LDPJoinSketch+} can be summarized as follows.
In phase 1, the server constructs the LDPJoinSketches using steps $\textcircled{1}$ and $\textcircled{2}$ based on the values from some sampled users. In step $\textcircled{3}$, the server computes the frequent items ($FI$) based on the sketches and sends this set to the remaining users.
In phase 2, for each attribute, the remaining users are divided into two groups: one for constructing a sketch for high-frequency items in $FI$ and another for low-frequency items. The $\textcircled{4}$ step involves encoding and perturbing each private value from  the two groups based on a Frequency-Aware-Perturbation (FAP) mechanism we proposed.  This mechanism can encode high-frequency and low-frequency items differently while satisfying LDP.  Finally, in step $\textcircled{5}$ the server separately constructs sketches to summarize the high-frequency items and low-frequency items.
Taking group $A_1$ for constructing the sketch $ML_A$ aiming at summarizing low-frequency items of attribute $A$ as an example, here low-frequency items $d_j\notin FI$ are target values and high-frequency ones $d_j\in FI$ are non-target values. FAP encodes each target value in the same way as $\textcircled{1}$, and encodes non-target items differently using a method denoted as $E_N$ in the figure. FAP ensures that the contributions of all non-target values can be removed from sketch $ML_A$, thereby reducing hash collisions based on the sketches separately summarizing high-frequency and low-frequency items.

The main contribution of this work:
\begin{itemize}
  \item We present a sketch-based join size estimation method called \textbf{LDPJoinSketch} for data with sensitive join values that have large domains. We prove that it satisfies LDP, and the estimation error is limited.
  \item We develop a \textbf{Frequency-Aware Perturbation (FAP)} mechanism that distinguishes between high-frequency and low-frequency items while ensuring compliance with LDP. Building upon FAP, we propose an enhanced method named \textbf{LDPJoinSketch+}, which reduces hash-collision errors while upholding privacy preservation.
  \item We conduct extensive experiments on synthetic and real-world datasets. The experimental results show that our methods outperforms state-of-the-art LDP mechanisms.
  \end{itemize}

\section{Related work}\label{sec:related_work}
\textbf{Join size estimation.}
The join size of two data streams is a crucial statistic in data stream analysis, representing the inner product of their frequency vectors. We mainly introduce the sketch-based join size estimation most relevant to this article, while the methods based on sampling and histograms have been well elaborated in previous works~\cite{Ganguly1996BifocalSF,Estan2006EndbiasedSF,Ioannidis1993OptimalHF}.
Sketch-based methods such as AGMS sketch~\cite{Alon1999TrackingJA} and fast-AGMS sketch~\cite{Cormode2005SketchingST} adopt compact probabilistic data structures to summarize data, and approximately compute the statistics~\cite{Chen2023DeepLB} such as frequency estimation and inner-product. 
To mitigate the error caused by hash collisions, studies such as Skimmed sketch~\cite{Ganguly2004ProcessingDJ}, Red sketch~\cite{Ganguly2005PracticalAF} and JoinSketch~\cite{Wang2023JoinSketchAS} have focused on separating high-frequency and low-frequency items. 

\textbf{Local differential privacy} (LDP) is a promising privacy preserving framework widely used in the field of private data releasing~\cite{Aydre2021DifferentiallyPQ}, data mining~\cite{Wang2018LocallyDP}, machine learning~\cite{Triastcyn2019BayesianDP} and social network analysis~\cite{Jiang2023ApplicationsOD}. 

Since no method has been proposed for estimating join size or inner product under LDP, we will instead focus on discussing the relevant task of frequency estimation under LDP. A basic LDP mechanism k-ary Randomized Response (k-RR)~\cite{Wang2017LocallyDP} perturbs the original value to any other values within the domain with the same probability. To handle large domain and reduce communication cost, Apple-HCMS~\cite{2017LearningWP} uses sketches to encode values before perturbation with hadamard transform. FLH~\cite{Cormode2021FrequencyEU} reduces the domain size by local hashing, which sacrifices accuracy to achieve computational gains.
These methods can be employed for estimating join sizes by calculating the frequency of each candidate join value. However, they suffer from both cumulative errors and efficiency issues, particularly when the join attribute has a large domain. In contrast, our LDPJoinSketch, which is based on the fast-AGMS sketch, estimates the join size by computing the product of sketches, rather than by accumulating the join size of each join value.

\section{Preliminaries}\label{sec:Preliminaries}

\subsection{AGMS sketch and Fast-AGMS sketch}\label{sec:fast-AGMS}
{AGMS sketch}~\cite{Alon1999TrackingJA} uses a counter $M_A$ to summarize the values of a data stream $A$. For every item $d$ in $A$, it adds $\xi(d)$ to the counter $M_A$, thus $M_A=\sum_{d \in A} \xi(d)$, where $\xi$ is a 4-wise independent hash function mapping $d$ to \{+1,-1\}. The join size of two data streams $A$ and $B$ can be estimated based on sketches $M_A$ and $M_B$ constructed with the same hash function $\xi$. The estimation is $Est = M_A \cdot M_B$. The accuracy can be enhanced by computing the median of multiple independent estimators.

{Fast-AGMS sketch}~\cite{Cormode2005SketchingST} is proposed to improve the construction efficiency of AGMS sketch. It adds a hash function $h$ to determine the location for an update, which avoids each update touching every counter as the AGMS does. A fast-AGMS sketch $M(k,m)$ is an array of counters of $k$ lines and $m$ columns. For each $j \in [k]$, two hash functions, $h_j$ and $\xi_j$, are employed for the $j$-th line of the sketch. Here, $h_j$ is responsible for determining which counter to increment, while $\xi_j$ is a member of a family of four-wise independent binary random variables uniformly distributed in $\{-1, +1\}$.
For each value $d$, it increases the counter with indices $[j,h_j(d)]$ by $M[j,h_j(d)]+=\xi_j(d)$. The join size $A\Join B$ can be estimated according to the following equation.
\begin{equation}
  Est=M_A\cdot M_B= \mathop{median}\limits_{i\in [1,k]}\{\sum_{j=1}^{m} M_A[i,j]\times M_B[i,j]\},
\end{equation}
where $M_A$ and $M_B$ are fast-AGMS sketches with parameters $(m,k)$ for $A$ and $B$, respectively.
The estimation error is limited as $\Pr[|M_A\cdot M_B-|A\Join B||>\frac{1}{\sqrt{m}}\|A\|_1\|B\|_1]\le\delta$, where $k=log(1/\delta)$, and $\|A\|_1$, $\|B\|_1$ are the number of values of attributes $A$ and $B$.


\subsection{Local differential privacy}
Local differential privacy (LDP)~\cite{Kasiviswanathan2008WhatCW} extends the notion of differential privacy (DP)~\cite{Dwork2006CalibratingNT} to scenarios where the aggregator cannot be trusted and
each data owner locally perturbs its data before sharing it with the server.

\begin{definition}
($\epsilon$-local differential privacy). A local randomized privacy algorithm $\mathcal{R}$ is $\epsilon$-locally differentially private($\epsilon$-LDP), where $\epsilon \geqslant 0$, iff for all pairs of inputs $x,x'\in \mathcal{D}$, we have that
\begin{equation}\label{LDP}
  \forall y \in \mathcal{Y}: Pr[\mathcal{R}(x)=y] \leqslant e^{\epsilon}\cdot Pr[\mathcal{R}(x')=y]
\end{equation}
where $\mathcal{Y}$ denotes the set of all possible outputs, and the privacy budget $\epsilon$ measures the level of privacy protection.
\end{definition}

\subsection{Hadamard Mechanism}
Hadamard Mechanism~\cite{2017LearningWP} utilizes hadamard transform to improve efficiency. The Hadamard Transform of a vector is obtained via multiplying with the hadamard matrix $H_m$. Here $H_m$ denotes the hadamard matrix of order $m$, a special type of square matrix where each element has only two possible values: +1 or -1. $H_m$ is defined recursively as
  $H_{m} =
  \begin{bmatrix}
    H_{m/2} & H_{m/2} \\
    H_{m/2} & -H_{m/2}
  \end{bmatrix}$, and $H_1=[1]$.
Hadamard transform spreads information from a sparse vector (1 in a single location across multiple dimensions), so that when we sample a bit from this dense set we still have sufficiently strong signal about the original vector.

\section{LDPJoinSketch}\label{sec:LDPJoinSketch}
We propose LDPJoinSketch, a local differentially private sketch for join size estimation. 
The LDPJoinSketch protocol can be divided into client-side and server-side algorithms. The goal of these two algorithms is as follows. \textbf{Client-side}: To perturb the private value and transmit the perturbed value along with the indices that determine which counter of the sketch in the server the perturbed value should be added to.
\textbf{Server-side}: To construct the sketch by adding each perturbed value to the counter with the given indices, and to compute the join size based on the constructed sketches.

\subsection{Client-side of LDPJoinSketch}\label{Sec:LDPJoinSketch-client}
The goal of the client-side of LDPJoinSketch is to encode and perturb each private value, ensuring that the output satisfies LDP and is safe to be transmitted to the server.

\begin{algorithm}
  \caption{Client-Side of LDPJoinSketch}\label{Algorithm:Client-Side}
  \hspace*{0.02in}{\textbf{Input}: join value $d\in D$, privacy budget $\epsilon$, sketch parameters ($k$, $m$), hash function pairs $\{(h_0,\xi_0),...,(h_{k-1},\xi_{k-1})\}$}\\
  \hspace*{0.02in}{\textbf{Output}: perturbed value $y$, indices $(j,l)$}
  \begin{algorithmic}[1]
    \State Sample $j$ uniformly at random from [k], sample $l$ uniformly at random from [m].
    \State Initialize a vector $v\leftarrow \{0\}^{1\times m}$
    \State Encode: $v[h_j(d)]\leftarrow \xi_j(d)$
    \State $w\leftarrow  v\times H_m$ \Comment{$H_m$: hadamard matrix of order $m$}
    \State Sample $b\in\{-1,+1\}$, which is -1 with probability $\frac{1}{e^{\epsilon}+1}$.
    \State Perturb: $y\leftarrow b\cdot w[l]$.
    \State \textbf{return}: $y$, ($j$, $l$)
  \end{algorithmic}
\end{algorithm}
The pseudo-code of the client-side is shown in Algorithm~\ref{Algorithm:Client-Side}. Given a private join value $d$, the privacy budget $\epsilon$, the sketch parameters $(k, m)$ indicating the number of lines and columns of the sketch in the server, the algorithm first samples a line index $j$ uniformly from $k$ lines and samples an index $l$ uniformly from $m$ columns (line 1). It then initializes a vector of size ($1\times m$) with all zeros. It then encodes the vector $v$ to be $\xi_j(d)$ in position $h_j(d)$ (line 3). To reduce communication cost by sampling while maintaining sufficiently strong signal about the original vector, we adopt hadamard transform before sampling as Apple-HCMS~\cite{2017LearningWP} does. The algorithm generates $w=v\times H_m$ which transforms $v$ with only one non-zero member $\xi_j(d)$ to a vector $w\in \{-\xi_j(d), \xi_j(d)\}^{1\times m}$. It then samples a bit $w[l]$ from the vector $w$, where $l\sim [m]$, and perturbs $w[l]$ by multiplying (-1) with the probability of $\frac{1}{e^{\epsilon}+1}$ (line 5-6). The client-side of LDPJoinSketch is almost the same as the client-side of Apple-HCMS~\cite{2017LearningWP}, and the only difference is the encoding method in line 3. We encode each value based on the fast-AGMS sketch and set  $v[h_j(d)]\leftarrow \xi_j(d)$, while Apple-HCMS encodes each item $d$ based on the Count Mean Sketch and sets $v[h_j(d)]\leftarrow 1$.


\begin{theorem}
  LDPJoinSketch satisfies $\epsilon$-LDP.
\end{theorem}
\begin{proof}
  Let $v$ and $v'$ be the encodings of $d$ and $d'$, respectively. According to Algorithm~\ref{Algorithm:Client-Side}, the differences between $v$ and $v'$ are on two bits, i.e., $v[h_j(d)]=\xi_j(d)$ and $v'[h_j(d')]=\xi_j(d')$. Let $J$ be the random variable selected uniformly from $[k]$, and $L$ be the random variable selected uniformly from $[m]$. Let $B$ be the random variable for the random bit $b$, where $\Pr[B=1]=\frac{e^\epsilon}{1+e^\epsilon}$ and $\Pr[B=-1]=\frac{1}{1+e^\epsilon}$. Let $w=v\times H_m$ and $w'=v'\times H_m$.
  We denote Algorithm~\ref{Algorithm:Client-Side} by $\mathcal{A}$ in the followings.
    \begin{align}\label{Equation:LDP-Proof1}
      &\frac{\Pr[\mathcal{A}(d)=(y,j,l)]}{\Pr[\mathcal{A}(d')=(y,j,l)]}\nonumber\\
      &=\frac{\Pr[B\cdot w[l]=y|J=j,L=l]\Pr[J=j]\Pr[L=l]}{\Pr[B\cdot w'[l]=y|J=j,L=l]\Pr[J=j]\Pr[L=l]}\nonumber\\
      &=\frac{\Pr[B\cdot H_m[h_j(d),l]\cdot \xi_j(d)=y|J=j]}{\Pr[B\cdot H_m[h_j(d'),l]\cdot \xi_j(d')=y|J=j]}
    \end{align}
  Since $\xi_j(d)$, $\xi_j(d')\in \{-1,1\}$, the samples from hadamard transform $(H_m[h_j(d),l]\cdot \xi_j(d))$, $(H_m[h_j(d'),l]\cdot \xi_j(d'))\in \{-1,1\}$. In addition, as $\Pr[B=1]=\frac{e^\epsilon}{1+e^\epsilon}$ and $\Pr[B=-1]=\frac{1}{1+e^\epsilon}$, the probability of obtaining the same output with different inputs $d$ and $d'$ is similar:
  \begin{equation}\label{Equantion:LDP-Proof2}
    e^{-\epsilon}\le\frac{\Pr[\mathcal{A}(d)=(y,j,l)]}{\Pr[\mathcal{A}(d')=(y,j,l)]}\le e^{\epsilon}
  \end{equation}
  Thus, LDPJoinSketch satisfies $\epsilon$-LDP.
\end{proof}

\subsection{Server-side of LDPJoinSketch}\label{sec:LDPJoinSketch-estimation}
After receiving the perturbed values, the server has two tasks: (1) to construct a sketch for each join attribute, and (2) to compute the join size based on the sketches.
\begin{algorithm}
  \caption{LDPJoinSketch-construction (PriSK)}\label{Algorithm:Server-Side}
  \hspace*{0.00in}{\textbf{Input}: $\{(y^{(1)},j^{(1)},l^{(1)})$, ...,$(y^{(n)},j^{(n)},l^{(n)})\}$, $\epsilon$, ($k$, $m$)}\\
  \hspace*{0.00in}{\textbf{Output}: private sketch $M$}
  \begin{algorithmic}[1]
    \State Initialize $M\in\{0\}^{k\times m}$.
    \State Set $c_\epsilon= \frac{e^\epsilon+1}{e^\epsilon-1}$
    \For {each $i\in [n]$}
    \State $M[j^{(i)},l^{(i)}]\leftarrow M[j^{(i)},l^{(i)}]+ k\cdot c_\epsilon \cdot y^{(i)}$
    \EndFor
    \State $M\leftarrow M\times H_m^T$
    \State \textbf{return: $M$}
  \end{algorithmic}
\end{algorithm}

\noindent\textbf{Construction of LDPJoinSketch.}
The pseudo-code of the construction algorithm is shown in Algorithm~\ref{Algorithm:Server-Side}. For each input $(y^{(i)},j^{(i)},l^{(i)})$ from the client-side, it multiplies $y^{(i)}$ with $k\cdot c_\epsilon$, and adds it to the counter at indices $[j^{(i)},l^{(i)}]$ (line 4). The scaling factor $k\cdot c_\epsilon$ (line 3) is used to debias the sketch, as both sampling an index $j$ and the perturbation in Algorithm 1 introduce bias: $\mathbb{E}[B]=\frac{1}{c_\epsilon}$ and $\mathbb{E}[\mathbbm{1}\{J=j\}]=1/k$. Finally, it multiplies the sketch $M$ with the hadamard matrix to transform the sketch back (line 6).

\begin{figure}[htbp]
  \centering
  \includegraphics[scale=0.55]{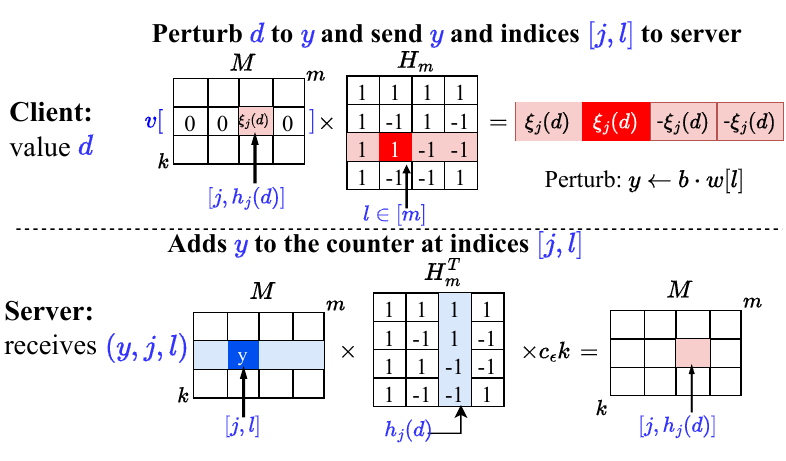}
  \caption{Example of LDPJoinSketch\vspace{-0.5cm}}
  \label{Fig:Example_of_LDPJoinSketch}
\end{figure}

\begin{example}\label{example1}
  We use an example in Fig.~\ref{Fig:Example_of_LDPJoinSketch} to demonstrate the working  of client-side and server-side in LDPJoinSketch. Suppose $k=3$, $m=4$ and  the private data in the client is $d$.

  \noindent\textbf{Client-side}: Let $h_j(d)=2$, where $j\in[k]$, the client encodes $d$ as $v=[0,0,\xi_j(d),0]$, where $v[h_j(d)]=\xi_j(d)$. It adopts hadamard transform to convert the only non-zero signal $\xi_j(d)$ to a vector  $v\times  H_m=[\xi_j(d),\xi_j(d),-\xi_j(d),-\xi_j(d)]$, and perturbs a random bit $w[l]$ of $w$ as $y=b\cdot w[l]$, where $l\in [m]$. Finally, the client sends $(y,j,l)$ to the server.

  \noindent\textbf{Server-side}: Construct the sketch by adding $y$ to the counter at indices $[j,l]$. It then multiplies the sketch with $H_m^T$ and a scale $c_\epsilon k$ to debias the sketch. We prove that the expectation of the contribution of $d$ to $M[j,h_j(d)]$ is still $\xi_j(d)$.
\end{example}

\noindent\textbf{Join size estimation.}
Based on the LDPJoinSketches $M_A$ and $M_B$ for attributes $A$ and $B$, the server estimates the join size:
\begin{equation}
  Est(|A\!\!\Join \!\!B|)\!=\!\!M_A M_B\!\!=\!\mathop{median}\limits_{j\in [1,k]}\{\sum_{x=1}^{m}\!\! M_A[j,x]M_B[j,x]\}
\end{equation}


\subsection{Analysis of estimation error}\label{sec:LDPJoinSketch-error}
We follow three steps to compute the error of join size estimation based on LDPJoinSketch.
Step 1, we analyze the contribution of each value $d$ to the sketch.
Step 2, we prove that each estimator  $M_A[j]\cdot M_B[j]$ ($j\in [k]$)  is an unbiased estimation of $|A\Join B|$. Additionally, we compute the variance of each estimator.
Step 3, we take the median of $k$ estimators as the final estimation and calculate the error bound for it.
For the convenience of the reader, we provide a notation table (Table~\ref{Tab:notation}) containing symbols in the following proofs.

\begin{table}[hptb]
  \centering
  \caption{The notations of symbols.}
  \label{Tab:notation}
  \begin{tabular}{ll}
  \hline
  Notations                                                       & Description \\ \hline
  $d$                                                             & private join attribute value\\
  $M$                                                             & the sketch \\
  $(k,m)$                                                             & the number of lines and columns of a sketch\\
  $h_{j}$                                                         & hash function for $j$th line, $h_j(x)\rightarrow [0,m-1]$\\
  $\xi_{j}$                                                       & 4-wise
  independent hash function $\xi_j(x)\rightarrow \{-1,+1\}$\\
  $f_A(d)$                                                      & true frequency of value $d$ of attribute $A$\\
  $M(j,h_j(d))^{(i)}$                                             & contribution of data  $d^{(i)}$ to the counter at $M[j,h_j(d)]$\\                                                                  
  $F_1(X)$                                                       & number of values of attribute $X$,i.e., $\sum_{d\in X}f_X(d)$\\
  $F_2(X)$                                                       & second frequency moment $\sum_{d\in X}(f_X(d))^2$ \\
  $L,R$                                                          & the random variables that select uniformly from $[m]$ \\
\hline
  \end{tabular}
  \end{table}

\noindent\textbf{Step 1. The contribution of each private value to the sketch.}
\begin{definition}
  The contribution of the $i$th private data $d^{(i)}$ to the counter at indices $(j,h_j(d))$ of the LDPJoinSketch $M$ sized $(k,m)$ can be written as $M(j,h_j(d))^{(i)}$.
  $M(j,h_j(d))^{(i)}=c_\epsilon k B\xi_j(d^{(i)}) H_m[h_j(d^{(i)}),L] H_m[L,h_j(d)]\mathbbm{1}\{J=j\}$, where $c_\epsilon=\frac{e^{\epsilon}+1}{e^{\epsilon}-1}$, and $\mathbbm{1}\{J=j\}$ is 0 when $J\neq j$. Here, $J\sim U[k]$, $L\sim U[m]$\footnote{$X\sim U[n]$ denotes $X$ is a variable chosen uniformly at random from [n]}, and $B\in\{-1,+1\}$ is 1 with probability $\frac{e^{\epsilon}}{1+e^{\epsilon}}$.
\end{definition}


In Theorem~\ref{Lemma:expectation_of_one_entry}, we give the expectation of $M(j,h_j(d))^{(i)}$.
\begin{theorem}\label{Lemma:expectation_of_one_entry}
  The contribution of a value $d^{(i)}$ to $M[j,h_j(d)]$:
  $\mathbb{E}[M(j,h_j(d))^{(i)}]\!=\!\xi_j(d)\cdot \mathbbm{1}\{d^{(i)}=d\}\! + \!\xi_j(d^{(i)})\frac{1}{m}\cdot \mathbbm{1}\{d^{(i)}\!\neq \!d\}$.
\end{theorem}
\begin{proof}
  We analyze $\mathbb{E}[ M(j,h_j(d))^{(i)}]$ under two situations.\\
  (1) If $d^{(i)}=d$,
  \begin{align}
     & \mathbb{E}[ M(j,h_j(d))^{(i)}]
      = \frac{1}{k}\mathbb{E}[ M(j,h_j(d))^{(i)}|J=j]                                       \\
     & = c_\epsilon\cdot \mathbb{E}[H_m[h_j(d),L]\xi_j(d^{(i)}) B\cdot H_m[L,h_j(d^{(i)})] ] \\
     & = c_\epsilon\cdot \mathbb{E}[\xi_j(d)B]\cdot \mathbb{E}[ H_m[L,h_j(d)]^2 ]
      = \xi_j(d)
  \end{align}
  (2) If $d^{(i)}\neq d$,
  \begin{align}\label{equation:expectation}
     & \mathbb{E}[ M(j,h_j(d))^{(i)}]
      =\frac{1}{k}\mathbb{E}[ M(j,h_j(d))^{(i)}|J=j]                                                   \\
     & =\mathbb{E}[\xi_j(d^{(i)})c_\epsilon B]\mathbb{E}[H_m[h_j(d^{(i)}),L] H_m[L,h_j(d)]]  \\
     & =\xi_j(d^{(i)})\mathbb{E}[\mathbbm{1}\{h_j(d)=h_j(d^{(i)})\}]=\xi_j(d^{(i)})\frac{1}{m}
  \end{align}
  Merging these two situations above, we can get: 
  \begin{equation}
  \mathbb{E}[M(j,h_j(d))^{(i)}]=\xi_j(d)\cdot \mathbbm{1}\{d^{(i)}=d\} + \frac{\xi_j(d^{(i)})}{m}\cdot \mathbbm{1}\{d^{(i)}\neq d\}.\nonumber
  \end{equation}
  The expectation is the same as that of fast-AGMS sketch.
\end{proof}
\noindent\textbf{Step 2: The expectation and variance of one estimator.}
Taking the inner-product of the corresponding lines of two sketches $M_A[j]\cdot M_B[j]$ as one estimator of $|A\Join B|$. We prove that $\mathbb{E}[M_A[j]M_B[j]]=|A\Join B|$ and  $\mathrm{Var}[M_A[j]M_B[j]]$ is limited.
Before computing $\mathbb{E}[M_A[j]M_B[j]]$, we first give a lemma for the product of two entries  from two datasets.

\begin{lemma}\label{Lemma:product_two_entries_from_two_sketches}
  Let $d^{(i_A)}$ and  $d^{(i_B)}$ be two values of attributes $A$ and $B$. Given $h_j(d^{(i_A)})=h_j( d^{(i_B)})=x$, we have
\begin{equation}
  \mathbb{E}[M_A(j,x)^{(i_A)}M_B(j,x)^{(i_B)}]=
  \begin{cases}
    1,&d^{(i_A)} = d^{(i_B)}\\
    0,&d^{(i_A)}\neq d^{(i_B)}
  \end{cases}
\end{equation}
\end{lemma}
\begin{proof}
  (1) If $d^{(i_A)} = d^{(i_B)}$, $\mathbb{E}[M_A(j,x)^{(i_A)}M_B(j,x)^{(i_B)}] =\mathbb{E}[\xi_j(d^{(i_A)})\xi_j(d^{(i_A)})]=1$.
  (2) If $d^{(i_A)}\neq d^{(i_B)}$, then $\mathbb{E}[M_A(j,x)^{(i_A)}\!M_B(j,x)^{(i_B)}]\!=\!\mathbb{E}[\xi_j(d^{(i_A)})\xi_j(d^{(i_B)})]\!=\!0$
\end{proof}
\begin{theorem}\label{Theorem:expectation}
    $\mathbb{E}[M_A[j]\cdot M_B[j]]=|A\Join B|$.
\end{theorem}
\begin{proof}
  $\mathbb{E}[M_A[j]\cdot M_B[j]]=\mathbb{E}[\sum_{x=1}^{m} M_A[j,x]\times M_B[j,x]]$.
  \begin{flalign}\label{Equation:UnbiasedResult}
     =&\sum_{x=1}^{m}\mathbb{E}[\sum_{d^{(i_A)}=d^{(i_B)}}M_A(j,x)^{(i_A)}M_B(j,x)^{(i_B)}\nonumber \\
       &+\sum_{d^{(i_A)}\neq d^{(i_B)}}M_A(j,x)^{(i_A)}M_B(j,x)^{(i_B)}] \nonumber \\
     =&\sum_{x=1}^{m}(\sum_{h_j(d)=x}f_A(d)f_B(d)+ 0)
     = |A\Join B|,
  \end{flalign}
 Here, $f_A(d)$ denotes the frequency of $d$ in attribute $A$.
\end{proof}
Before computing $\mathrm{Var}[M_A[j]M_B[j]]$, we provide Lemma~\ref{lemma:product_in_one_sketch} computing $\mathbb{E}[M(j,x)^{(i_1)}M(j,x)^{(i2)}]$ and  Lemma~\ref{lemma:second_moment_of_one_entry_one_sketch} computing  $\mathbb{E}[(M[j,l])^2]$ based on Lemma~\ref{lemma:product_in_one_sketch} as follows.

\begin{lemma}\label{lemma:product_in_one_sketch}
  Given $h_j(d^{(i_1)})= h_j(d^{(i_2)})=x$, we have 
  \begin{equation}
    \!\!\!\!\mathbb{E}[M(j,x)^{(i_1)}M(j,x)^{(i_2)}]=
    \begin{cases}
      c_{\epsilon}^2k, &i_1=i_2\\
      1,&\!\!\!\!\!\!\!\!i_1\ne i_2, d^{(i_1)}=d^{(i_2)}\\
      0,&\!\!\!\!\!\!\!\!i_1\ne i_2, d^{(i_1)}\ne d^{(i_2)}
    \end{cases}
  \end{equation}
\end{lemma}
\begin{proof}
  According to the definition of $M(j,x)^{(i)}$, given two data $d^{(i_1)}$ and $d^{(i_2)}$ of the same attribute, we have
  $\mathbb{E}[M(j,x)^{(i_1)}M(j,x)^{(i_2)}]=c_{\epsilon}^2k^2\cdot \mathbb{E}[J^{(i_1)}\!\!\!\!=\!\!\!\!J^{(i_2)}\!\!=\!\!\!\!j]\cdot\mathbb{E}[\xi_j(d^{(i_1)})\xi_j(d^{(i_2)})B^{(i_1)}B^{(i_2)}]$,
   where $J^{(i_1)},J^{(i_2)}\sim U[k]$ represent random variables chosen uniformly from $[k]$ for $d^{(i_1)}$ and $d^{(i_2)}$. Similarly, $B^{(i_1)}$ and $B^{(i_2)}$ represent the random variables for perturbation.\\
  Case 1. If $i_1=i_2=i$, then $d^{(i_1)}$ and $d^{(i_2)}$ represent the same data entry $d^{(i)}$, so $\mathbb{E}[J^{(i_1)}=J^{(i_2)}=j]=\frac{1}{k}$ and $B^{(i_1)}B^{(i_2)}=(B^{(i)})^2=1$. Therefore, $\mathbb{E}[M(j,x)^{(i_1)}M(j,x)^{(i_2)}]=c_{\epsilon}^2k$.\\
  Case 2. If $i_1\ne i_2$, then $J^{(i_1)}$ and $J^{(i_2)}$ are independent variables, $\mathbb{E}[J^{(i_1)}\!\!=\!\!J^{(i_2)}\!\!\!=\!\!\!j]\!\!=\!\!\mathbb{E}[J^{(i_1)}\!\!\!\!=\!\!j]\mathbb{E}[J^{(i_2)}\!\!\!\!=\!\!j]\!\!=\!\!1/{k^2}$. Similarly $\mathbb{E}[B^{(i_1)}B^{(i_2)}]=\mathbb{E}[B^{(i_1)}]^2=\frac{1}{c_{\epsilon}^2}$. And $\xi_j(d^{(i_1)})\xi_j(d^{(i_2)})=1$, since $d^{(i_1)}=d^{(i_2)}$. Therefore, $\mathbb{E}[M(j,x)^{(i_1)}M(j,x)^{(i_2)}]=1$.\\
  Case 3. If $d^{(i_1)}\ne d^{(i_2)}$, similar to (2), we have $\mathbb{E}[J^{(i_1)}=J^{(i_2)}=j]=\frac{1}{k^2}$ and  $\mathbb{E}[B^{(i_1)}B^{(i_2)}]=\frac{1}{c_{\epsilon}^2}$. The difference from Case 2 is that $E[\xi_j(d^{(i_1)})\xi_j(d^{(i_2)})]=0$ for  $d^{(i_1)}\ne d^{(i_2)}$.
\end{proof}


\begin{lemma}\label{lemma:second_moment_of_one_entry_one_sketch}
  The expectation of  $(M[j,x])^2$:
\begin{equation}
  \mathbb{E}[(M[j,x])^2]
   =\!\!\!\!\!\!\sum_{h_j(d)=x}\!\!\!\! (kc_{\epsilon}^2-1)f(d)+f(d)^2
   +\!\!\!\!\!\!\!\!\!\!\sum_{\substack{d\neq d'\\h_j(d)=h_j(d')=x}}\!\!\!\!\!\!\!\!\!\!\!\xi_j(d)\xi_j(d')\nonumber
\end{equation}
\end{lemma}
\begin{proof}
  Suppose $h_j(d^{(i_1)})= h_j(d^{(i_2)})=x$, where $d^{(i_1)}$ and $d^{(i_2)}$ are two data entries for sketch $M$. Since $M[j,x]=\sum_{d^{(i)}\in D}M(j,h_j(d))^{(i)}$, based on lemma~\ref{lemma:product_in_one_sketch}, we have
  \begin{align}
    &\mathbb{E}[(M[j,x])^2]=\!\!\!\!\!\!\!\!\!\sum_{d^{(i_1)},d^{(i_2)}\in D} \!\!\!\!\!\!\!\!\!\mathbb{E}[(M(j,x)^{(i_1)})(M(j,x)^{(i_2)})] \nonumber  \\
  = &\!\!\!\!\!\!\!\!\!\sum_{d^{(i_1)},d^{(i_2)}\in D}\!\!\!\!\!\!\!\!\![c_{\epsilon}^2k \mathbbm{1}\{i_1=i_2\}  + \mathbbm{1}\{\substack{i_1\ne i_2\\d^{(i_1)}=d^{(i_2)}}\} \nonumber  \\
  &+\xi_j(d^{(i_1)})\xi_j(d^{(i_2)})\mathbbm{1}\{\substack{i_1\ne i_2\\d^{(i_1)}\ne d^{(i_2)}}\}]\\
  =& kc_{\epsilon}^2\!\!\!\!\sum_{h_j(d)=x}\!\!\!\!f(d)+\!\!\!\!\sum_{h_j(d)=x}\!\!\!\!f(d)^2\!\!-\!\!f(d)
   +\!\!\!\!\!\!\!\!\!\!\!\!\sum_{\substack{d\neq d'\\h_j(d)=h_j(d')=x}}\!\!\!\!\!\!\!\!\xi_j(d)\xi_j(d')
   \end{align}
   \noindent Based on this, we can  compute  $\mathrm{Var}[M_A[j,x]M_B[j,x]]$.
\end{proof}

We present herein a definition for certain notations that will be utilized in the subsequent theorems.

\begin{definition}\label{Def:frequency moment}
  Let $d$ be a value of attribute $X$ and $f_X(d)$ be the frequency (number of occurrences) of $d$ in the sequence of values of attribute $X$. We define (1) the total frequency of all the values of $X$ as  $F_1(X) = \sum_{d \in X} f_X(d)$, (2) the second frequency moment $F_2(X) = \sum_{d \in X} (f_X(d))^2$.
\end{definition}

\begin{lemma}\label{Lemma:variance_of_one_product}
  $\mathrm{Var}[M_A[j,x]M_B[j,x]]
 \le\frac{2}{m^2}(F_1(A)+\frac{kc_{\epsilon}^2-1}{2})^2\times (F_1(B)+\frac{kc_{\epsilon}^2-1}{2})^2$.
\end{lemma}
\begin{proof}
  Let $X=M_A[j,x]M_B[j,x]$, $\mathrm{Var}[X]=\mathbb{E}[X^2]-\mathbb{E}[X]^2$. 
  (i) First, we compute $\mathbb{E}[X^2]$.
  \begin{align}
    M_A[j,x]^2=&\frac{1}{m}[(kc_{\epsilon}^2-1)F_1(A)+F_2(A)\nonumber\\
    &+\!\!\!\!\!\!\sum_{\substack{d_A\ne d'_A\\h_j(d_A)=h_j(d'_A)}}\!\!\!\!\!\!f_A(d_A)f_A(d'_A)\xi_j(d_A)\xi_j(d'_A)]\label{equation:sqrt_one_cell}
  \end{align}
  Based on Eq.~(\ref{equation:sqrt_one_cell}), we have $\mathbb{E}[X^2]= \mathbb{E}[M_A[j,l]^2M_B[j,l]^2]$
  \begin{align}
    &=\frac{1}{m^2}[(kc_{\epsilon}^2-1)F_1(A)+F_2(A)][(kc_{\epsilon}^2-1)F_1(B)+F_2(B)]\nonumber\\
    &+\frac{4}{m^2}\sum_{d < d'}f_A(d)f_A(d')f_B(d)f_B(d')
   \end{align}
  (ii) Second, $\mathbb{E}[X]=\frac{1}{m}\sum_{d}f_A(d)f_B(d)$.\\
  Based on (i) and (ii), we have $\mathrm{Var}[M_A[j,x]M_B[j,x]]$
  \begin{align}\label{Equation:variance}
   \le&\frac{1}{m^2}[(kc_{\epsilon}^2-1)F_1(A)+F_2(A)][(kc_{\epsilon}^2-1)F_1(B)+F_2(B)]\nonumber\\
   &+\frac{1}{m^2}F_2(A)F_2(B)\\
  \le &\frac{2}{m^2}(F_1(A)+\frac{kc_{\epsilon}^2-1}{2})^2(F_1(B)+\frac{kc_{\epsilon}^2-1}{2})^2,
 \end{align}
 where $F_2(A)$ and $F_2(B)$ are defined in Def~\ref{Def:frequency moment}.
\end{proof}
With Lemma~\ref{Lemma:variance_of_one_product}, we compute $\mathrm{Var}[M_A[j]M_B[j]]$ as follows.
\begin{theorem}
  The variance of $M_A[j]M_B[j]$ is limited:\\
 $\mathrm{Var}[M_A[j]M_B[j]]
 \le \frac{2}{m}(F_1(A)+\frac{kc_{\epsilon}^2-1}{2})^2\times (F_1(B)+\frac{kc_{\epsilon}^2-1}{2})^2$.
\end{theorem}
\begin{proof}
  The variance of each estimator $M_A[j]M_B[j]$:
  \begin{align}
      &\mathrm{Var}[M_A[j]M_B[j]] = \mathrm{Var}[\sum_{x=1}^{m} M_A[j,x] M_B[j,x]]     \\
     &\le\frac{2}{m}(F_1(A)+\frac{kc_{\epsilon}^2-1}{2})^2(F_1(B)+\frac{kc_{\epsilon}^2-1}{2})^2
  \end{align}
  Thus, variance of each estimator is bounded.
\end{proof}

\noindent\textbf{Step 3. Error bound of LDPJoinSketch.}
The final estimation is computed as the median of $k$ estimators like $Est_j$, i.e., $Est=\mathop{median}\limits_{j\in[k]}(Est_j)$. We use the following theorem to prove that the error of $Est$ is limited.
\begin{theorem}\label{theorem:error_bound}
  Let $k=4\log\frac{1}{\delta}$, $Est=\mathop{median}\limits_{j\in[k]}Est_j$, and the join size  estimation error be $Er=Est-|A\Join B|$,
  \begin{equation}
  \Pr[|Er|\ge\frac{4}{\sqrt{m}}|F_1(A)+\frac{kc_{\epsilon}^2-1}{2}||F_1(B)+\frac{kc_{\epsilon}^2-1}{2}|]\le \delta\nonumber.
  \end{equation}
\end{theorem}

\begin{proof}
According to Chebyshev's Inequality, given a random variable $X$, $Pr[|X-E[X]|\ge w]\le\frac{\mathrm{Var}[X]}{w^2}$.
We can compute the error bound for each estimator $Est_j=M_A[j]M_B[j]$,
\begin{equation}
  \Pr[|Est_j-|A\Join B||\ge w]\le\frac{Var[M_A[j]M_B[j]]}{w^2}
\end{equation}
Let $w=\sqrt{8Var[M_A[j]M_B[j]]}$, we get
\begin{equation}
  \Pr[|Est_j-|A\Join B||\ge \sqrt{8\mathrm{Var}[M_A[j]M_B[j]]}]\le\frac{1}{8}
\end{equation}
Using the application of Chernoff Bounds, let $k=4\log\frac{1}{\delta}$, we can use $Est=\mathop{median}\limits_{j\in[k]}(Est_j)$ 
to reduce the failure probability to $\delta$.
$\Pr[|Er|\ge\frac{4}{\sqrt{m}}|F_1(A)+\frac{kc_{\epsilon}^2-1}{2}||F_1(B)+\frac{kc_{\epsilon}^2-1}{2}|]\le \delta$.
\end{proof}

Based on Theorem~\ref{theorem:error_bound}, the error of the LDPJoinSketch-based join size estimation is limited. LDP introduces an additional error to the error bound, but the influence caused by LDP does not destroy the utility of the sketch, since $\frac{kc_{\epsilon}^2-1}{2}$ is much smaller than $F_1(A)$ and $F_1(B)$ in reality.


\section{LDPJoinSketch+}\label{sec:LDPJoinSketchPlus}

To reduce the hash-collision while preserving privacy, we present a novel two-phase framework known as LDPJoinSketch+. This framework builds upon and enhances the accuracy of LDPJoinSketch by specifically addressing hash collisions (Section~\ref{sec:framework_LDPJoinSketchPlus}). In LDPJoinSketch+, We design a Frequency-Aware Perturbation (FAP) mechanism for the client-side (Section~\ref{sec:Frequency_Aware_Perturbation}). Finally, we show how to estimate the join size based on the LDPJoinSketch+ (Section~\ref{sec:server_LDPJoinSketchPlus}).




\subsection{The framework of LDPJoinSketch+}\label{sec:framework_LDPJoinSketchPlus}
LDPJoinSketch+ aims to reduce error caused by hash-collisions under LDP. It has a two-phase framework.
In phase 1, we identify candidate frequent join values from each table and communicate these candidates to the clients. In phase 2, each client distinguishes whether its value is high-frequency or not, and encodes high-frequency and low-frequency values in distinct manners, all while ensuring compliance with LDP.

We use the pseudo-code in Algorithm~\ref{algorithm:framework-LDPJoinSketch+} to show the framework of LDPJoinSketch+. LDPJoinSketch+ has two phases.
In phase 1, we find the frequent item set $FI$ based on the LDPJoinSketches $M_A$ and $M_B$ constructed from sample users for attributes $A$ and $B$ (Phase 1, line 1-4).
In phase 2, to avoid allocating privacy budget, we divide the users into two groups (Phase 2, line 1), one for the join size estimation of low-frequency values, another for high-frequency values. We can use the whole privacy budget for each group according to the composition theorem of LDP. For ease of description, we encapsulate the process of perturbing each data through FAP and aggregating all perturbed data to construct a sketch as a function called $sk$. Based on $sk$, we can construct sketches $ML_A$ and $ML_B$ summarizing the low-frequency items of attribute $A$ and $B$ with the data from $A_1$ and $B_1$ (Phase 2, line 2). $MH_A$ and $MH_B$ summarizing the high-frequency items can be constructed similarly (Phase 2, line 3).
It then computes the join size for low-frequency values as $LEst$ and for high-frequency values as $HEst$ (line 4-5). The final result is the sum of the scaled $HEst$ and $LEst$ (Phase 2, line 6). Because $LEst$ is the estimated join size of data from $A_1$ and $B_1$, we estimate the join size $|A\Join B|$ by multiplying a scale $\frac{|A||B|}{|A1||B1|}$.

\begin{algorithm}
  \caption{LDPJoinSketch+}\label{algorithm:framework-LDPJoinSketch+}
  \hspace*{0.02in} {\bf Phase 1: Find frequent join values}
  \begin{algorithmic}[1]
  \State $S_A, S_B\leftarrow$ sample clients from $A$ and $B$ respectively.
  \State Clients: Perturb data from $S_A,S_B$ with Algorithm 1.
  \State Server: Construct LDPJoinSketch $M_A$ and $M_B$.
  \State Frequent item set $FI\leftarrow FreqItems(MA,MB)$.
\end{algorithmic}
  \hspace*{0.02in} {\bf Phase 2: Join size estimation}
  \begin{algorithmic}[1]
  \State Groups $A_1$, $A_2\leftarrow A$; Groups $B_1$, $B_2\leftarrow B$.
  \State $ML_A\leftarrow sk(A_1,L,\epsilon,FI)$; $ML_B\leftarrow sk(B_1,L,\epsilon,FI)$.
  \State $MH_A\leftarrow sk(A_2,H,\epsilon,FI)$; $MH_B\leftarrow sk(B_2,H,\epsilon,FI)$.
  \State $LEst\leftarrow JoinEst(ML_A,ML_B, mode=L)$.
  \State $HEst\leftarrow JoinEst(MH_A,MH_B, mode=H)$.
  \State \textbf{return: $\frac{|A||B|}{|A_1||B_1|}LEst+\frac{|A||B|}{|A_2||B_2|}HEst$}
  \end{algorithmic}
  \hspace*{0.02in} {\bf Func $sk(D, mode, \epsilon, FI)$}
  \begin{algorithmic}[1]
    \For {$d^{(i)}\in D$}
      \State $(y^{(i)}, j^{(i)}, l^{(i)})\leftarrow FAP(d^{(i)}, mode,\epsilon,FI)$ 
    \EndFor
    \State $M\leftarrow PriSk((y^{(1)}, j^{(1)}, l^{(1)}),\dots,(y^{(n)}, j^{(n)}, l^{(n)}),\epsilon,k,m)$
    \State return $M$
  \end{algorithmic}
\end{algorithm}

We will introduce FAP mechanism and how to estimate the join size based on the LDPJoinSketch+ in subsequent sections.

\subsection{Frequency-Aware Perturbation}\label{sec:Frequency_Aware_Perturbation}
Phase 2 of LDPJoinSketch+ aims to improve accuracy by estimating the join size of high-frequency and low-frequency items separately. To achieve this without privacy leaks, we propose a frequency-aware perturbation (FAP) mechanism.

As the phase 1 of LDPJoinSketch+ gets the frequent items, this knowledge enables each client to distinguish whether its value is high-frequency or not, and differently handle high-frequency and low-frequency items. For ease of explanation, we refer to the value that possesses the same property as the estimation target as the ``target value'', while the value that does not possess the same property as the estimation target will be referred to as ``non-target values''. For instance, in constructing sketches for estimating the join size of high-frequency items, the high-frequency values would be considered target values, while the low-frequency values would be considered non-target values. FAP has two goals: (1) the server cannot infer whether the true value is high-frequency or not from the perturbed value, (2) the contribution of non-target values to the sketches can be removed.

For the first goal, FAP encodes the target values in the same way as LDPJoinSketch does, while encoding non-target values independently from their true values. Both the encoded values of target and non-target values are perturbed according to random response as LDPJoinSketch does before being sent to the server. Thus, the privacy still can be preserved.

For the second goal, non-target values are encoded randomly, resulting in their impact spreading uniformly across each cell of the sketch. Consequently, the influence of non-target values can be effectively removed from the target sketch, given that we know the total number of non-target values. This improves the accuracy of target estimation.

\begin{algorithm}
  \caption{Frequency-Aware Perturbation (FAP)}\label{Algorithm:FAP}
  \hspace*{0.02in} {\textbf{Input}: $d\in D$, mode, $\epsilon$,  $FI$}\\
  \hspace*{0.02in} {\textbf{Output}: perturbed value $y$, index $j,l$}
  \begin{algorithmic}[1]
  \If  { (mode==H) == ($d\notin FI$)} \Comment{Non-target value}
    \State Sample $j$ uniformly at random from [k]
    \State Sample $l,r$ uniformly at random from [m].
    \State Initialize a vector $v\leftarrow \{0\}^{1\times m}$
    \State Set $v[r]\leftarrow 1$
    \State Transform $w\leftarrow v\times H_m$
    \State Sample $b\in\{-1,+1\}$, where $\Pr[b=-1]=\frac{1}{e^{\epsilon}+1}$.
    \State $y\leftarrow b\cdot w[l]$.
  \Else \Comment{Target value}
    \State $y,j,l \leftarrow$ LDPJoinSketch-client($d$, $\epsilon$, $m$, $k$)
  \EndIf
  \State \textbf{return: $y,j,l$}
  \end{algorithmic}
\end{algorithm}

The pseudo-code for the FAP algorithm is presented in Algorithm~\ref{Algorithm:FAP}. The parameter $mode$ signifies the encoding model for data. When $mode=H$, it indicates that the target values are high-frequency ones. In this case, low-frequency items are encoded randomly. Conversely, when $mode=L$, the target values are low-frequency ones. Therefore, when both $mode==H$ and $d\notin FI$ are either true or false, the algorithm employs random encoding for data $d$. Otherwise, it encodes and perturbs data $d$ using the client-side of LDPJoinSketch.

The sole distinction between encode methods for target and  non-target values lies in line 5: FAP encodes a non-target value $d$ randomly with $v[r]\leftarrow 1$ ($r$ is chosen uniform at random from $[m]$), but encodes a target value $d$ by $v[h_j(d)]\leftarrow \xi_j(d)$ according to the client side of LDPJoinSketch (line 10).

\begin{example}\label{example2}
  We also use an example shown in Fig.~\ref{Fig:Example_of_LDPJoinSketchPlus} to show the process of FAP. Since a target value is encoded and perturbed based on LDPJoinSketch, we only show how to handle a non-target value. The parameters in this example are the same with those in Example~\ref{example1}.
  A significant distinction from LDPJoinSketch is that FAP employs a randomly chosen variable $r\in [m]$ to replace the index $h_j(d)$. Consequently, the encoding of a non-target value $d$ becomes independent of its true value, ensuring that the impact of non-target values is evenly distributed across each cell of the sketch.
\end{example}

\begin{figure}[htbp]
  \centering
  \vspace{-0.5cm}
  \includegraphics[scale=0.55]{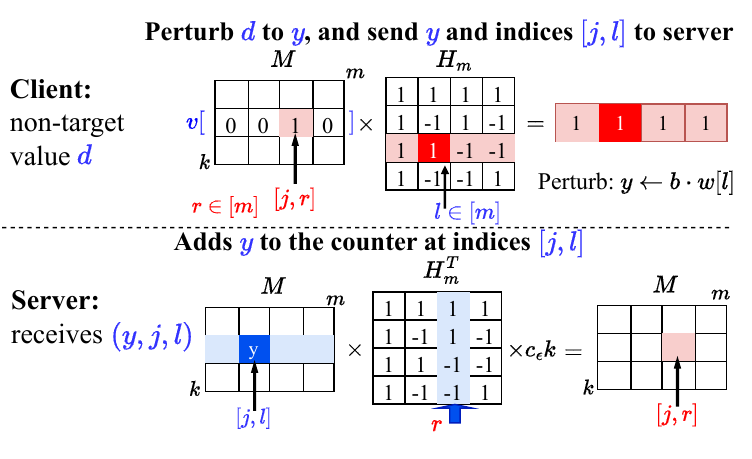}
  \caption{Example of LDPJoinSketch+}
  \label{Fig:Example_of_LDPJoinSketchPlus}
\end{figure}


\begin{theorem}
  Algorithm FAP satisfies $\epsilon$-LDP.
\end{theorem}
\begin{proof}
  We have proved that LDPJoinSketch satisfies $\epsilon$-LDP. 
  We only need to prove that the server still cannot distinguish the outputs of a non-target value $d$ and a target value $d'$.

  Suppose the encode of $d$ and $d'$ are $v$ and $v'$, respectively, the differences between $v$ and $v'$ are on two bits, i.e., $v[r]=1$ and $v'[h_j(d')]=\xi_j(d')$. Let $J$ be the random variable that selects uniformly from  $[k]$, $L$ be the random variable that selects uniformly from $[m]$, and $B$ be the random variable for the random bit $b$. Here, $\Pr[B=1]=\frac{e^\epsilon}{1+e^\epsilon}$, $\Pr[B=-1]=\frac{1}{1+e^\epsilon}$.
  \begin{equation*}
     \frac{\Pr[\mathcal{A}(d)=(y,j,l)]}{\Pr[\mathcal{A}(d')=(y,j,l)]} \\
      =\frac{\Pr[B H_m[l,r]=y|J=j]}{\Pr[B H_m[l,h_j(d')]\xi_j(d')=y|J=j]}
  \end{equation*}
  Since both $H_m[l,r], H_m[l,h_j(d')]\xi_j(d')\in \{-1,1\}$, the probability of obtaining the same output with different inputs $d$ and $d'$ is similar, as demonstrated in the following equation.
  \begin{equation}
    e^{-\epsilon}\le\frac{\Pr[\mathcal{A}(d)=(y,j,l)]}{\Pr[\mathcal{A}(d')=(y,j,l)]}\le e^{\epsilon}
  \end{equation}
  Thus, the Algorithm FAP, denoted as $\mathcal{A}$, satisfies $\epsilon$-LDP.
\end{proof}

Each client perturbs its value based on FAP. The server receives the perturbed values from two groups and construct sketches for join size estimation of high-frequency and low-frequency values respectively. We will show how to compute the join size based on LDPJoinSketch+ in the next part.

\subsection{The server-side of LDPJoinSketch+}\label{sec:server_LDPJoinSketchPlus}
The server-side of LDPJoinSketch+ has different aims in the two phases. Phase 1 aims to find the frequent items. Phase 2 aims to construct sketches summarizing high-frequency and low-frequency values separately and estimate the join size. 

\noindent\textbf{Phase 1: Frequency estimation based on LDPJoinSketch.}\\
We use the following Theorem~\ref{Theorem:FrequencyEstimation} to prove that the LDPJoinSketch can provide unbiased frequency estimations.
\begin{theorem}\label{Theorem:FrequencyEstimation}
  Given an LDPJoinSketch $M$ summarizing values of attribute $A$. The frequency of a value $d\in A$  can be estimated as $\tilde{f}(d)=\mathop{mean}\limits_{j\in [1,k]} M[j,h_j(d)]\xi_j(d)$.
Here, $\tilde{f}(d)$ is the unbiased estimation of $f(d)$, i.e., $\mathbb{E}[\tilde{f}(d)]=f(d)$.
\end{theorem}

\begin{proof}
  According to lemma~\ref{Lemma:expectation_of_one_entry}, we have $M(j,h_j(d))^{(i)}=\xi_j(d)$ for $d^{(i)}=d$, and $M(j,h_j(d))^{(i)}=\frac{\xi_j(d^{(i)})}{m}$ for $d^{(i)}\ne d$.
  \begin{align*}
    &\mathbb{E}[\tilde{f}(d)]= \frac{1}{k}\sum_{j=1}^{k}\sum_{i=1}^{n}\mathbb{E}[M(j,h_j(d))^{(i)}\xi_j(d)]\nonumber\\
    &=\frac{1}{k}\sum_{j=1}^{k}[f(d)\xi_j(d)\xi_j(d)+\!\!\!\!\!\!\!\!\!\sum_{\substack{d'\ne d\\h_j(d)=h_j(d')}}\!\!\!\!\!\!\!\!\!f(d')\xi_j(d')\xi_j(d)]\!=\!f(d)
    \end{align*}
   LDPJoinSketch provides unbiased frequency estimations.
\end{proof}

As LDPJoinSketch can provide unbiased frequency estimation, given a proper threshold $\theta$ to separate the high-frequency and low-frequency values, we can find the frequent item set $FI_A=\{d_i\in A|\tilde{f}(d_i)>\theta |A|\}$, and $FI_B=\{d_i\in B|\tilde{f}(d_i)>\theta|B|\}$ for the attributes $A$ and $B$, where $|A|$ and $|B|$ represent the total frequency of items in attributes $A$ and $B$, respectively. We let the final frequent item set be the union of $FI_A$ and $FI_B$, i.e., $FI=FI_A\cup FI_B$.

\noindent\textbf{Phase 2: Join size estimation.}

The server-side of LDPJoinSketch+ constructs the sketches $MH_A, ML_A$ for attribute $A$ and $MH_B, ML_B$ for attribute $B$ in the same way as LDPJoinSketch. So we only discuss how to estimate the join size based on LDPJoinSketch+.

LDPJoinSketch+ (Algorithm 3) calls $JoinEst$ to separately compute the join size estimation for high-frequency and low-frequency values.
The pseudo-code for the $JoinEst$ algorithm is presented in Algorithm~\ref{Algorithm:estimate}. The algorithm initially calculates the total frequencies of elements belonging to the high-frequency sets of attributes $A$ and $B$ (lines 1-3). Different parameters $mode$ represent different target values for the sketch. When $mode=L$, it indicates that the sketches $M_A$ and $M_B$ aim to summarize low-frequency values. In this case, the algorithm removes the contribution of high-frequency values from the sketches (lines 5-8). Conversely, when $mode=H$, it removes the contribution of low-frequency values from the sketches (lines 9-12).

  \begin{algorithm}
    \caption{JoinEst}\label{Algorithm:estimate}
    \hspace*{0.02in}{\textbf{Input}: $M_A$, $M_B$, mode}\\
    \hspace*{0.02in}{\textbf{Output}: Join size estimation $Est$}
    \begin{algorithmic}[1]
      \For {$d \in FI$} \Comment{$FI$ is the frequent item set in phase 1.}
        \State $HighFreq_A+=\tilde{f}_A(d)\cdot \frac{|A|}{|SA|}$
        \State $HighFreq_B+=\tilde{f}_B(d)\cdot \frac{|B|}{|SB|}$
      \EndFor
      \If {mode==L}
        \State $M_A\leftarrow M_A- \{\frac{HighFreq_A}{m}\}^{k\times m}$
        \State $M_B\leftarrow M_B- \{\frac{HighFreq_B}{m}\}^{k\times m}$
        \State $Est = M_A \cdot M_B$.
      \Else { mode==H}
        \State $M_A\leftarrow M_A- \{\frac{|A|-HighFreq_A}{m}\}^{k\times m}$
        \State $M_B\leftarrow M_B- \{\frac{|B|-HighFreq_B}{m}\}^{k\times m}$
        \State $Est = M_A \cdot M_B$.
      \EndIf
      \State \textbf{return: $Est$}
    \end{algorithmic}
  \end{algorithm}

  The following theorem computes the contribution of non-target values to sketch $M$.
  \begin{theorem}
    The contribution of non-target values to $M[j,x]$:
    $\mathbb{E}[\sum_{d^{(i)}\in [NT]}M(j,x)^{(i)}]=|NT|/m$,
  where $|NT|$ is the total frequency of all the non-target values.
  \end{theorem}
  \begin{proof}
    The contribution of a non-target value $d^{(i)}$ to $M[j,x]$:
    \begin{align}
       &\mathbb{E}[M(j,x)^{(i)}] = \frac{1}{k}\mathbb{E}[M(j,x)^{(i)}|J=j]  \nonumber \\
       & = \mathbb{E}[c_\epsilon H_m[x,L]H_m[L,R]B]
       = \mathbb{E}[\mathbbm{1}\{x=R\}]
        = \frac{1}{m},
    \end{align}
    where $L,R$ are random variables that select uniformly from $[m]$. Therefore, $\mathbb{E}[\sum_{d^{(i)}\in [NT]}M(j,x)^{(i)}]=|NT|/m$.
  \end{proof}

With the above theorem, we can remove the contribution of non-target values from the sketches. For example, $|NT|$ is $HighFreq_A$ for sketch $M_A$ when $mode=L$ (line 6). In this way, we separately estimate the join size of high-frequency and low-frequency items.

\section{Extension to multi-way Joins}\label{sec:Extension}
Taking inspiration from COMPASS~\cite{DBLP:conf/sigmod/IzenovDRS21}, which utilizes multi-dimensional fast-AGMS sketches to facilitate multi-way joins, we illustrate in Fig.~\ref{Fig:multijoin} how our LDPJoinSketch can be extended to support multi-way joins.

\begin{figure}
  \centering
  \includegraphics[width=0.485\textwidth]{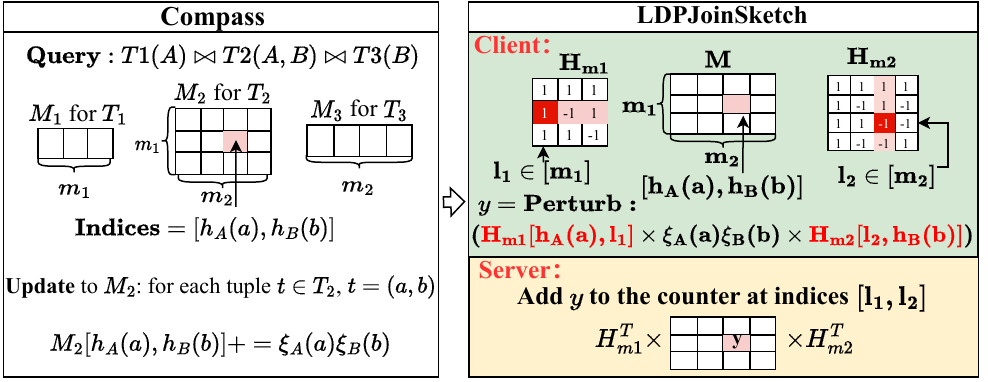}
  \caption{LDPJoinSketch for multi-way join.}
  \label{Fig:multijoin}
\end{figure}

Taking $Q=T_1(A)\mathop{\Join}T_2(A,B) \mathop{\Join}T_3(B)$ as a multi-way join example. In this scenario, each join attribute is equipped with a unique pair of hash functions denoted as $h$ and $\xi$. COMPASS constructs fast-AGMS sketches $M_1$ and $M_3$ for attributes $T_1.A$ and $T_3.B$ respectively. As for $T_2$, it constructs a 2-dim matrix $M_2$ in shape of ($m_1\times m_2$). For every tuple $t\in T_2$ with $t.A=a$ and $t.B=b$, the counter at indices $[h_A(a),h_B(b)]$ is incremented by $\xi_A(a)\xi_B(b)$. The size of $Q$ can be estimated as:
$\sum\limits_{\substack{l_1\in [m_1],l_2\in [m_2]}}M_{1}[l_1]\cdot M_{2}[l_1,l_2]\cdot M_{3}[l_2]$.
The accuracy can be improved by computing the median of $k$ independent estimators like the one in this example.
Our LDPJoinSketch can support multi-way joins in a similar way.
We can directly construct LDPJoinSketches $\tilde{M}_{1}$ and $\tilde{M}_{3}$ for $T_1$ and $T_3$ which contain only one join attribute in each table. Therefore, we only discuss how to construct a sketch $M_{2}$ for $T_2$ with two join attributes. 
Assuming each attribute has only one-pair hash functions $h$ and $\xi$, for a tuple $t\in T_2$ with $t.A=a$ and $t.B=b$, we encode $t$ using hadamard matrixes $H_{m_1}$ and $H_{m_2}$. We extract the $h_A(a)$th line (pink cells) of $H_{m_1}$, the $h_B(b)$th column (pink cells) of $H_{m_2}$, and randomly select two indices  $l_1\in[m_1]$ and $l_2\in[m_2]$.
The client-side encodes $t$ by $H_{m_1}[h_A(a),l_1]\times\xi_A(a)\xi_B(b)\times H_{m_2}[l_2,h_B(b)]$, and then perturbs the encoded value by multiplying it by (-1) with probability $\frac{1}{e^\epsilon+1}$. The perturbed value $y$ is then sent to the server to update the counter with indices $[l_1,l_2]$. The server constructs the sketch with the indices and perturbed values, and restores the sketch by $\tilde{M}=H_{m_1}^T\cdot M \cdot H_{m_2}^T$. A scale factor $c_\epsilon$ is applied to debias the sketch, giving $\tilde{M}\leftarrow c_\epsilon \tilde{M}$. Finally, the server computes the estimation:
\begin{equation}
 Est=\sum\limits_{\substack{l_1\in [m_1],l_2\in [m_2]}}\tilde{M}_{1}[l_1]\cdot \tilde{M}_{2}[l_1,l_2]\cdot \tilde{M}_{3}[l_2]
\end{equation}

\textbf{Discussion.}
We give a solution for simple chain multi-join operations. For chain multi-join queries on tables with a maximum of two join attributes in each table, the computational complexity\footnote{The computational complexity of matrix multiplication is $O(m^3)$, which can be further reduced to $ O(m^{2.371552})$~\cite{williams2024new}.} of our method is $O(m^3)$, if each two-dimensional sketch has a shape of $(m\times m)$. The complexity increases for scenarios involving star join and cyclic join operations. While our method adeptly handles uncomplicated cyclic joins, such as $ T_1(A,B) \Join T_2(B,C) \Join T_3(C,A) $, addressing general multi-joins poses challenges due to intricate join graphs and the need for privacy budget decomposition among multiple join attributes. These challenges exceed the scope of this paper, and addressing them within a single work is a formidable task. Therefore, we defer this exploration to future research.

\section{Experiments}\label{sec:Experiments}

In this section we design experiments to evaluate the performance of our methods on synthetic and real datasets. 

\subsection{Experimental Setup}

\noindent \underline{\textbf{Hardware.}}
We implement all the algorithms on a machine of 256 GB RAM running Ubuntu(20.04.1) with Python 3.9.

\noindent \underline{\textbf{Queries.}}
We consider the following form of join queries:\\
    Q: Select Count(*) from $T_1$ join $T_2$ on $T_1.A=T_2.B$,
where $A$ and $B$ are the private join attributes of tables $T_1$ and $T_2$.

\noindent \underline{\textbf{Datasets.}}
We generate one-dimensional synthetic datasets following zipf and gaussian distribution, respectively. Regarding these synthetic data as the join attribute values is the common setting in previous works~\cite{Ganguly2004ProcessingDJ, Wang2023JoinSketchAS}. We also conduct the experiments on four real-world datasets.

(1) \textbf{Zipf datasets}: We generate several datasets of size $N$ following Zipf distribution with different skewness parameters, whose probability mass function is $f(x|\alpha,N) = \frac{1/x^\alpha}{\sum_{n=1}^{N}(1/n^\alpha)}$, where $\alpha$ is the skewness parameter and $x$ is the rank of item.

(2) \textbf{Gaussian dataset}: We also generate a dataset that follows Gaussian distribution, whose probability density function for a given value $x$ is $f(x) = \frac{1}{\sigma \sqrt{2\pi}} e^{-\frac{(x-\mu)^2}{2\sigma^2}}$, where $\mu$ is the mean of distribution and $\sigma$ is the standard deviation. 

(3) \textbf{TPC-DS dataset}\footnote{https://www.tpc.org/tpcds/}: The TPC Benchmark™ DS is a benchmark for measuring the performance of decision support systems. We extract the store sales data for experiments. 

(4) \textbf{MovieLens dataset}\footnote{https://grouplens.org/datasets/movielens/}. The MovieLens dataset is commonly used in the field of recommender systems and collaborative filtering, containing movie ratings and user information. 

(5) \textbf{Twitter ego-network dataset}\footnote{https://snap.stanford.edu/data/ego-Twitter.html}: This real-world dataset consists data of circles from Twitter.

(6) \textbf{Facebook ego-network}\footnote{https://snap.stanford.edu/data/ego-Facebook.html}: Facebook data is collected from survey participants using Facebook app.

Table~\ref{Table:datasets_info} shows the domain and data size of all datasets.

\begin{table}[hptb]
    \centering
    \caption{Information of Datasets.}
    \begin{tabular}{|c|c|c|}
      \hline
      \textbf{Dataset Name} & \textbf{Domain} & \textbf{Size} \\
      \hline
      Zipf & 4,377--2,816,390 & 40,000,000 \\
      Gaussian & 75,949 & 40,000,000 \\
      MovieLens & 83,239 & 67,664,324 \\
      TPC-DS & 18,000 & 5,760,808 \\
      Twitter & 77,072 & 4,841,532 \\
      Facebook & 4,039 & 352,936 \\
      \hline
    \end{tabular}
    \label{Table:datasets_info}
\end{table}


\noindent \underline{\textbf{Competitors}}

(1) \textbf{k-RR}~\cite{Wang2017LocallyDP}. We adopt it to perturb the join values and compute the join size using the calibrated frequency vectors.

(2) \textbf{Apple-HCMS}~\cite{2017LearningWP}. Apple's private hadamard count mean sketch, which also forms a $k\times m$ sketch.

(3) \textbf{FLH}~\cite{Cormode2021FrequencyEU}. The heuristic fast variant of Optimal Local Hashing (OLH).

(4) \textbf{Fast-AGMS}~\cite{Cormode2005SketchingST}. It estimates the join size \textbf{without privacy-preserving}. We refer to this method as \textbf{FAGMS} for short in the figures.



\noindent \underline{\textbf{Error Metrics}}



(1) \textbf{Absolute Error (AE)}. 
$\frac{1}{t} \sum \vert \mathcal{J} - \mathcal{\hat{J}} \vert$, where $\mathcal{J}$, $\mathcal{\hat{J}}$ denote the true and estimated  join size, and $t$ is the testing rounds.

(2) \textbf{Relative Error (RE)}. 
$\frac{1}{t} \sum \vert \mathcal{J} - \mathcal{\hat{J}} \vert /\mathcal{J}$. The parameters are the same as those defined in AE.

(3) \textbf{Mean Squared Error (MSE)}. Metric for frequency estimation. $MSE=\frac{1}{n} \sum_{d \in D}(f(d)-\tilde{f}(d))^2$, where $f(d)$ and $\tilde{f}(d)$ are true and estimated frequencies of $d$, and $n$ is the number of distinct values in $D$.

\noindent \underline{\textbf{Parameters}}

$\epsilon$: The privacy budget of differential privacy. 

$\alpha$: The skewness parameter of synthetic Zipf datasets. 

$(k, m)$: The number of lines and columns of a sketch.

$r$: Sampling rate of LDPJoinSketch+ used in phase 1.

$\theta$: The threshold for separating high-frequency items and low-frequency items.

\begin{figure}[htbp]
  \centering
  \includegraphics[scale=0.4, trim={0 0 0 0}, clip]{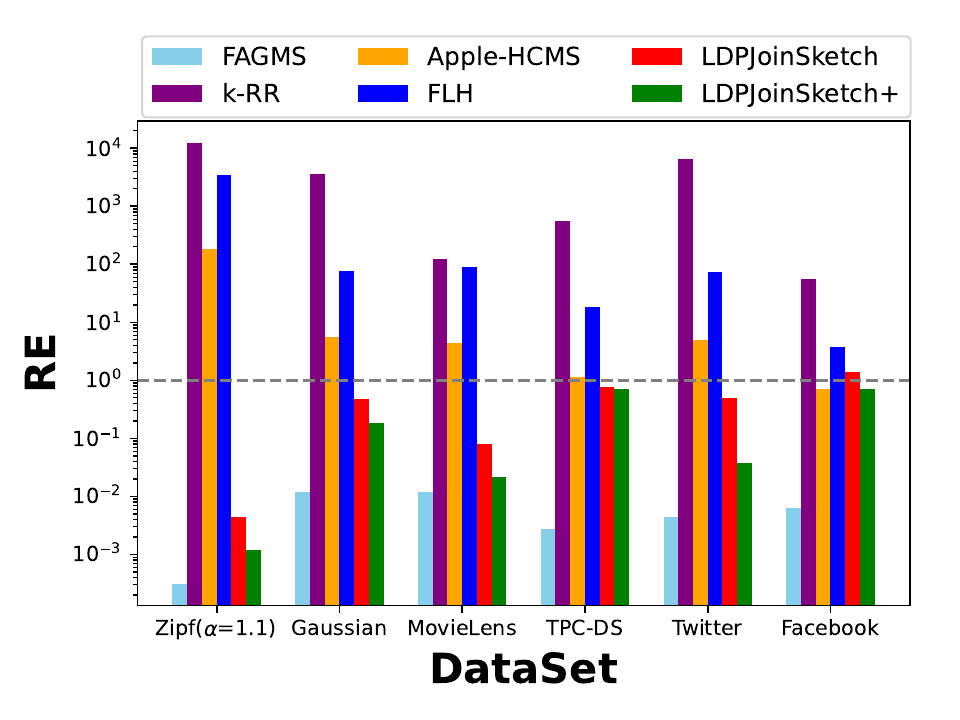}
  \caption{The Accuracy of Join Size Estimation.}
  \label{Fig:Accuracy}
\end{figure}

\subsection{Accuracy}
We test the accuracy of different methods on all the six datasets, and the experimental results are illustrated in Fig.~\ref{Fig:Accuracy}. We fix the privacy parameter $\epsilon=4$ and set sketch parameters ($k=18$, $m=1024$) for sketch-based methods. 
Our proposed method outperforms others in most cases, and the relative error (RE) of our methods are sufficiently small, akin to those of FAGMS(Non-private), which indicates that it provides rigorous privacy protection while maintaining excellent accuracy.  Because methods like k-RR and FLH introduce significant noises to the join value, consequently leading to extremely poor results. Apple-HCMS does not consider hash-collisions. Besides, we can observe that LDPJoinSketch+ achieves better utility to a certain extent on different datasets. The reason is that the FAP further reduces the hash-collision error. Our method demonstrates clear advantages in situations with large data volumes and large domains, but it does not exhibit significant advantages in small datasets with small domain, cause most LDP mechanisms are not suitable for small data. LDP requires a large amount of data to ensure that the noise introduced by perturbation satisfies the desired distribution.

\subsection{Space and Communication cost}
\textbf{Space Cost}. We test the accuracy of sketch-based methods with similar space cost on Zipf($\alpha=2.0$) dataset and the results are shown in Fig.~\ref{Fig:SpaceCost}. We set $\epsilon=10$, $r=0.1$, and $\theta=0.001$. The space cost of Apple-HCMS and LDPJoinSketch only contains the size of one sketch for each table. The space cost of LDPJoinSketch+ includes the size of sketches used in two phases. 
For simplicity, we set same size of sketches used in two phases in LDPJoinSketch+, which means that the space cost in phase 2 is nearly twice that of phase 1. We take a variety of settings to make different methods have similar space costs.
We can observe that our LDPJoinSketch+ attains greater accuracy compared to Apple-HCMS at a comparable space cost level. 


\begin{figure}[hptb]
  \centering
  \begin{minipage}{0.24\textwidth}
    \centering
  \includegraphics[width=\linewidth, trim={0 0 1cm 1cm}, clip]{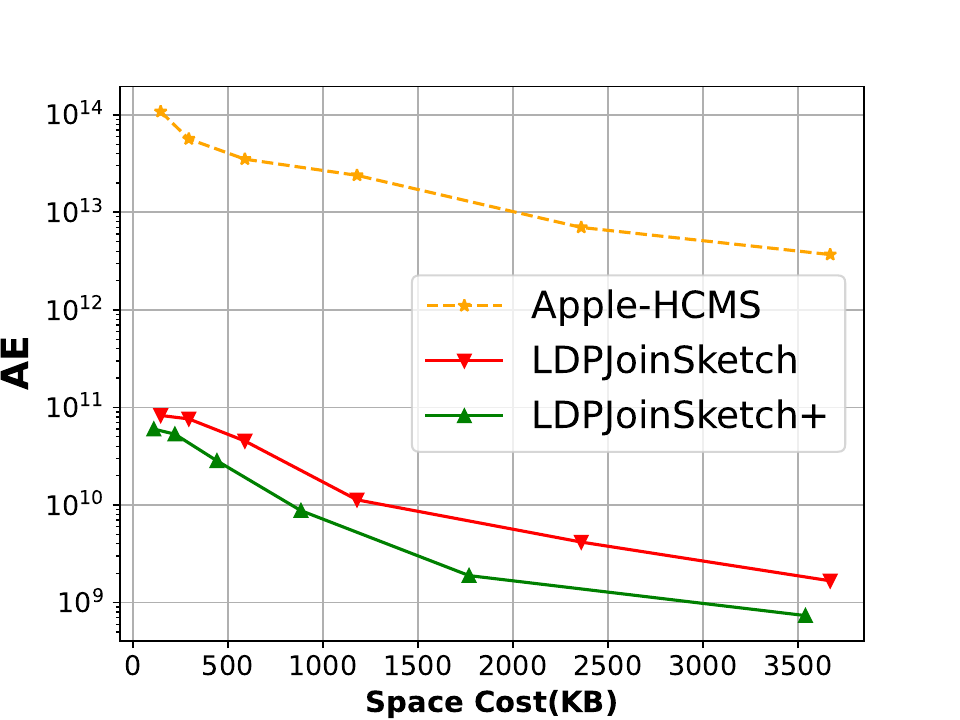}
  \caption{Impact of space cost.}
  \label{Fig:SpaceCost}
  \end{minipage}
  \hfill
  \begin{minipage}{0.24\textwidth}
    \centering
  \includegraphics[width=\linewidth, trim={0 0 1cm 1cm}, clip]{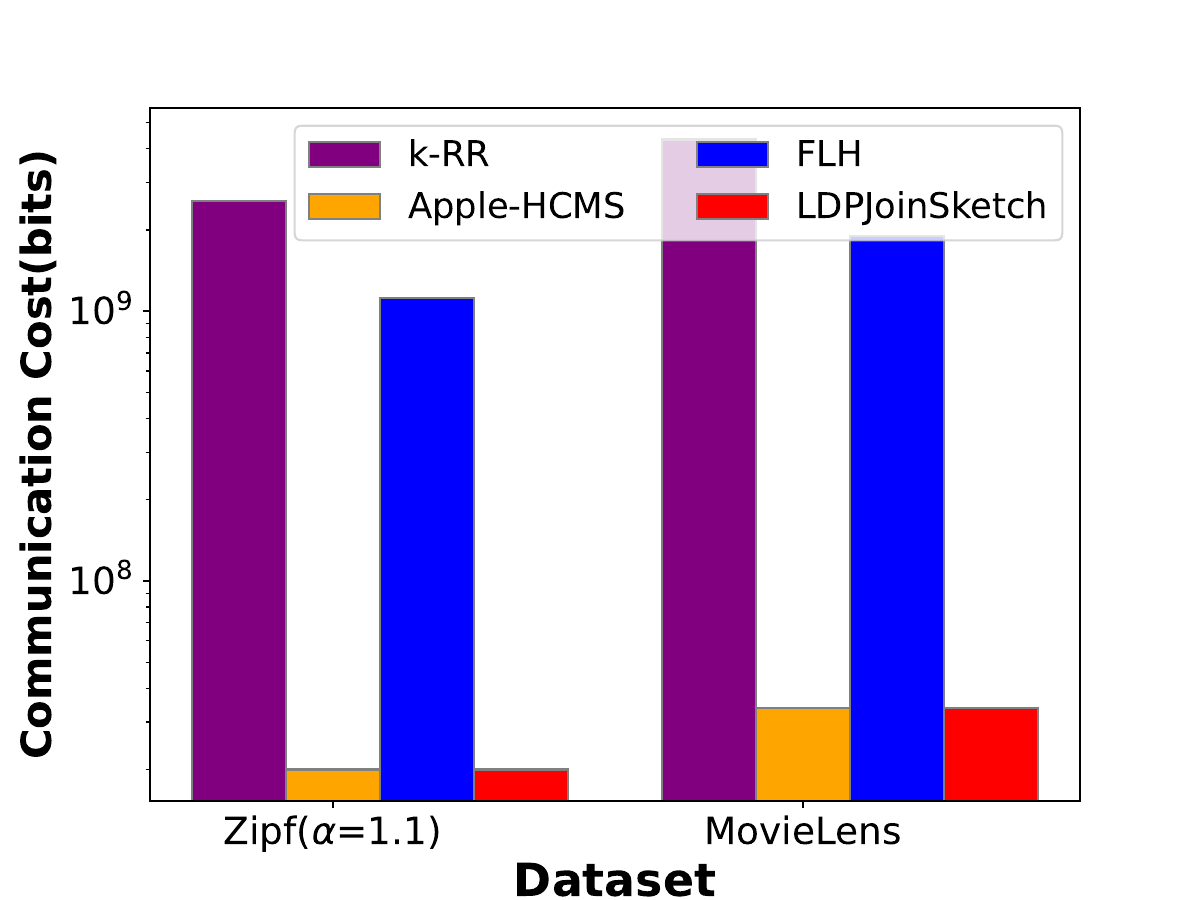}
  \caption{Communication cost.}
  \label{Fig:Communication}
  \end{minipage}
\end{figure}


\textbf{Communication cost.} Fig.~\ref{Fig:Communication} shows the communication costs of different methods on Zipf($\alpha=1.1$) and MovieLens datasets. We set ($k=18$, $m=1024$) and $\epsilon=4$.  The y-axis represents the cumulative number of bits transmitted from all clients to the server. Since LDPJoinSketch and Apple-HCMS both encode each value by the hadamard matrix and send only one bit to the server, they can significantly reduce the communication cost. Instead, k-RR sends the complete encoded value to the server, and FLH also sends an encoded vector with optimal length to the server.


\begin{figure*}[htbp]
    \centering
    \subfigure[Zipf($\alpha$=1.5)]{
        \centering
        \begin{minipage}[b]{0.22\textwidth}
	    \includegraphics[width=1\textwidth, trim={0.2cm 0 0.61cm 0}, clip]{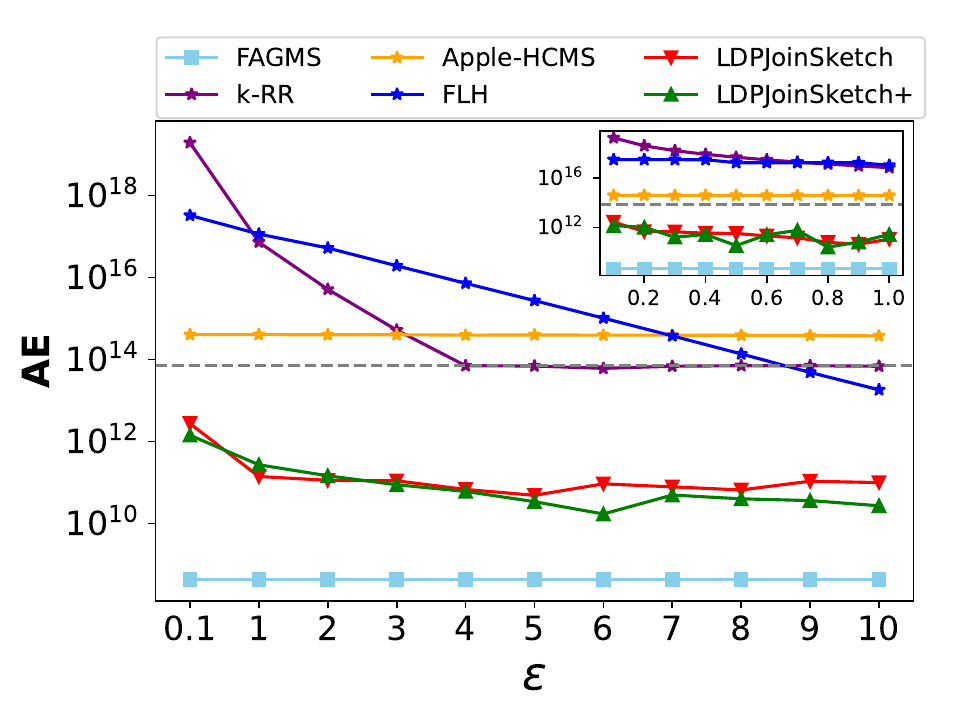}
	    \end{minipage}
	    \label{fig:eps_zipf}
    }
    \subfigure[Gaussian]{
        \centering
        \begin{minipage}[b]{0.22\textwidth}
	    \includegraphics[width=1\textwidth, trim={0.2cm 0 0.61cm 0}, clip]{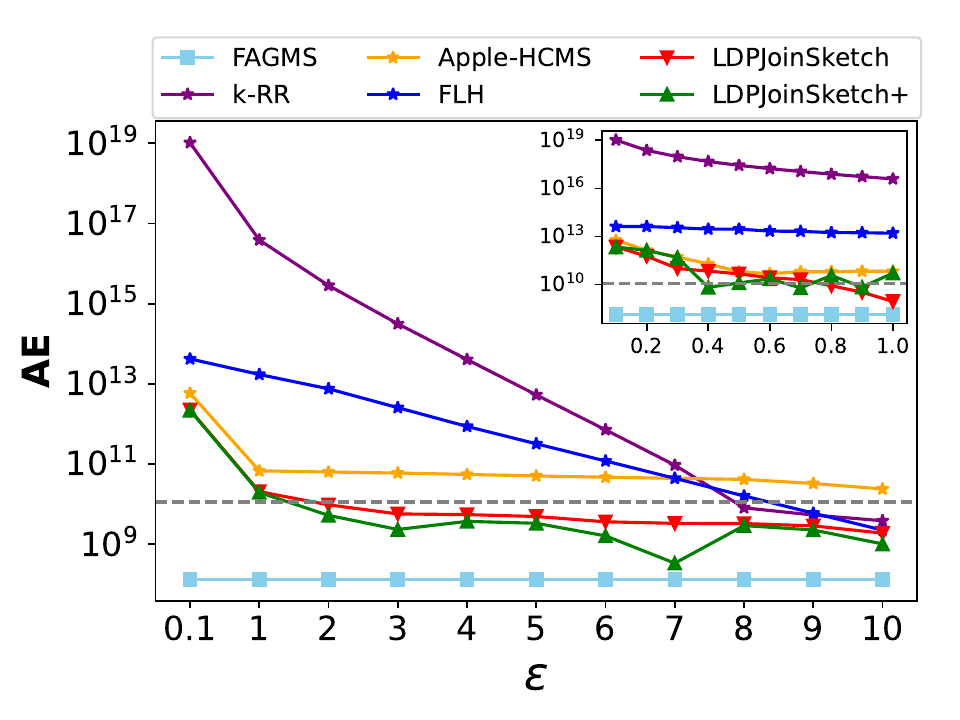}
	    \end{minipage}
	    \label{fig:eps_gaussian}
    }
    \subfigure[MovieLens]{
        \centering
        \begin{minipage}[b]{0.22\textwidth}
	    \includegraphics[width=1\textwidth, trim={0.2cm 0 0.61cm 0}, clip]{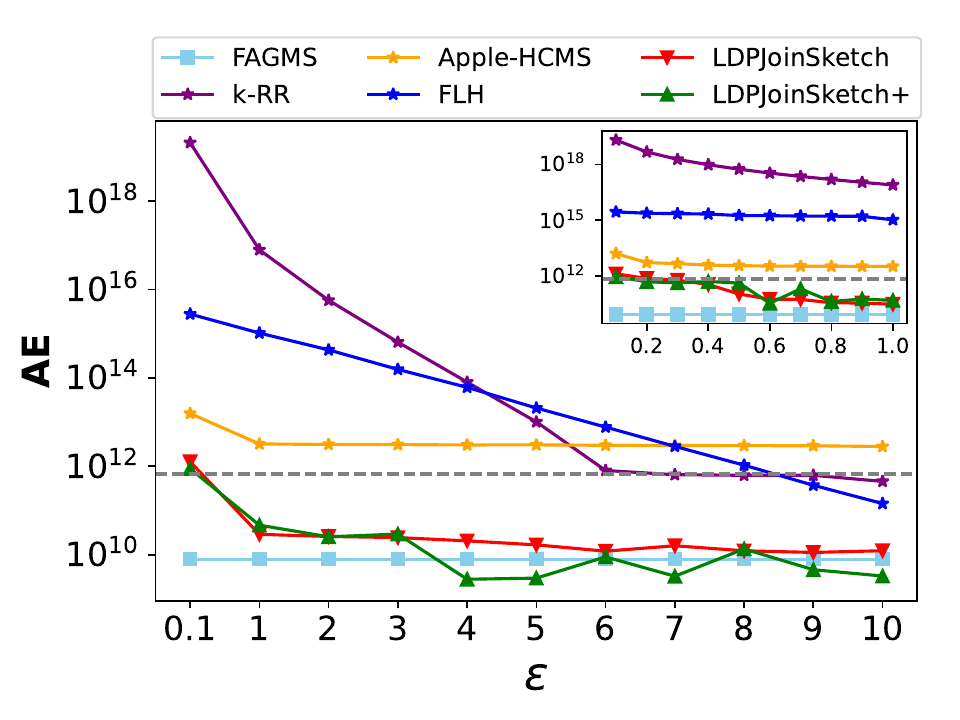}
	    \end{minipage}
	    \label{fig:eps_movielen}
    }
    \subfigure[Twitter]{
        \centering
        \begin{minipage}[b]{0.22\textwidth}
	    \includegraphics[width=1\textwidth, trim={0.2cm 0 0.61cm 0}, clip]{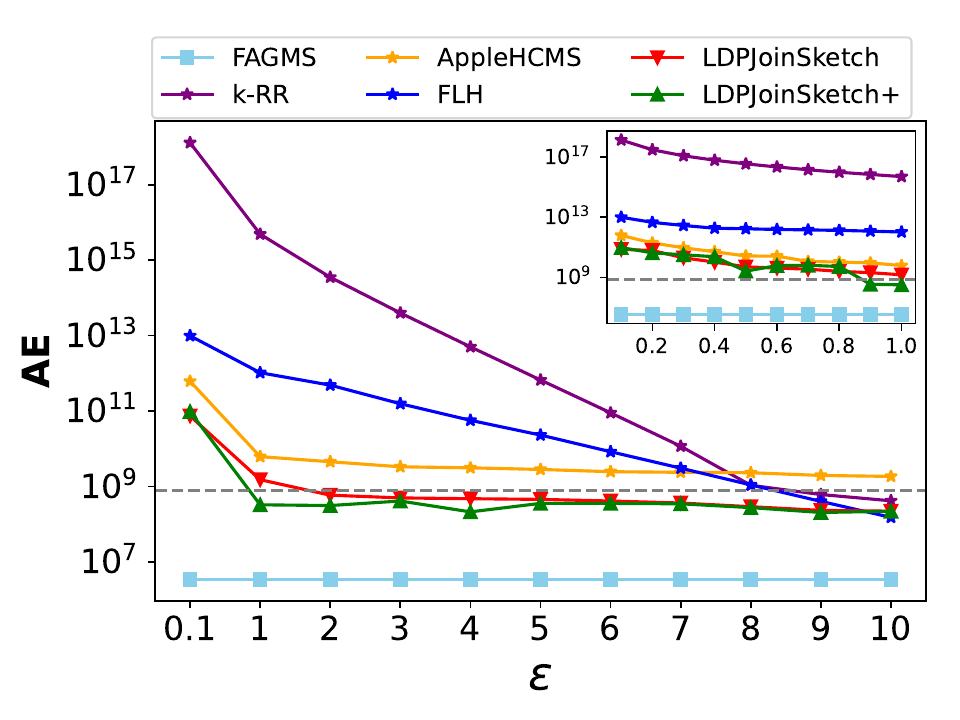}
	    \end{minipage}
	    \label{fig:eps_twitter}
    }
       \caption{Impact of $\epsilon$.}
	\label{fig:epsilon_impact}
\end{figure*}


\begin{figure*}[htbp]
  \centering
  \subfigure[Zipf($\alpha$=1.1)]{
        \centering
        \begin{minipage}[b]{0.22\textwidth}
      \includegraphics[width=1\textwidth, trim={0 0 1.2cm 1cm}, clip]{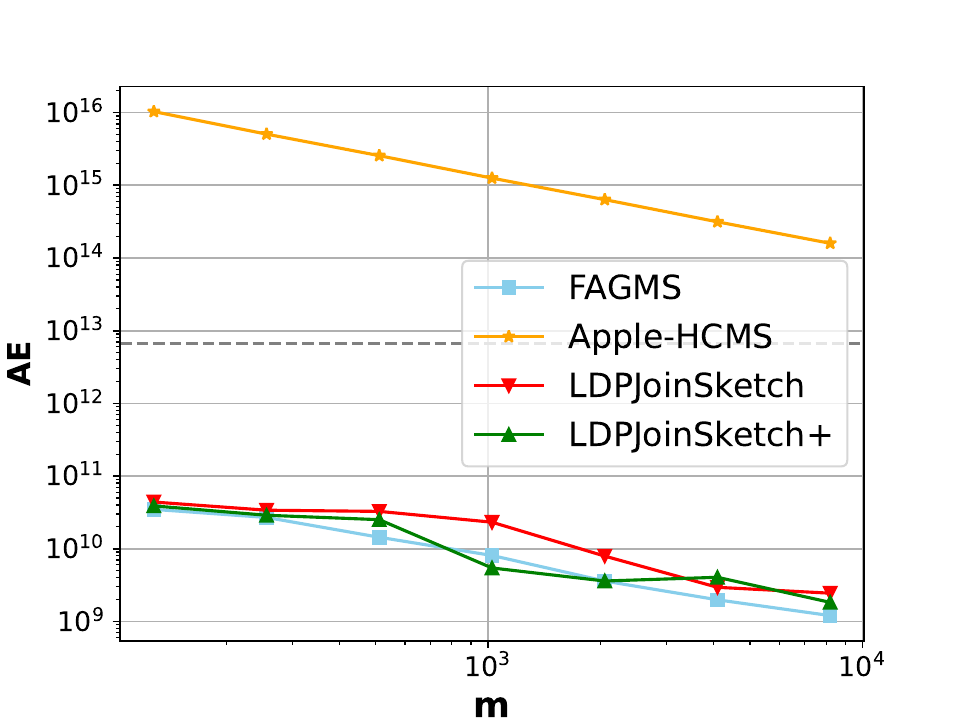}
      \end{minipage}
      \label{fig:m_zipf1_1}
    }
    \subfigure[Zipf($\alpha$=2.0)]{
        \centering
        \begin{minipage}[b]{0.22\textwidth}
      \includegraphics[width=1\textwidth, trim={0 0 1.2cm 1cm}, clip]{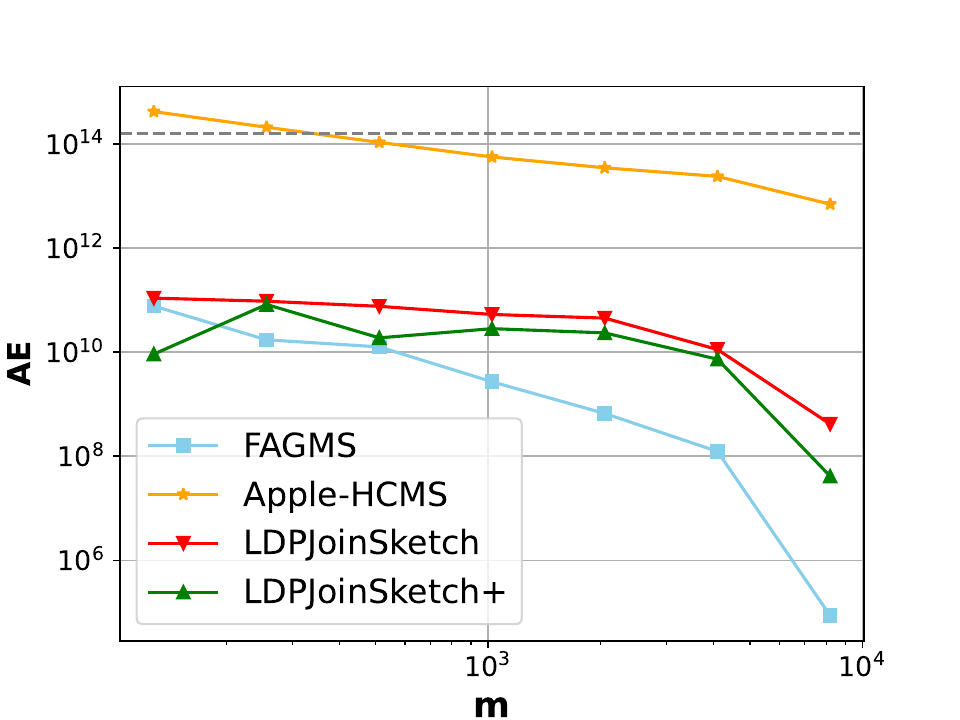}
      \end{minipage}
      \label{fig:m_zipf1_5}
    }
    \subfigure[MovieLens]{
        \centering
        \begin{minipage}[b]{0.22\textwidth}
      \includegraphics[width=1\textwidth, trim={0 0 1.2cm 1cm}, clip]{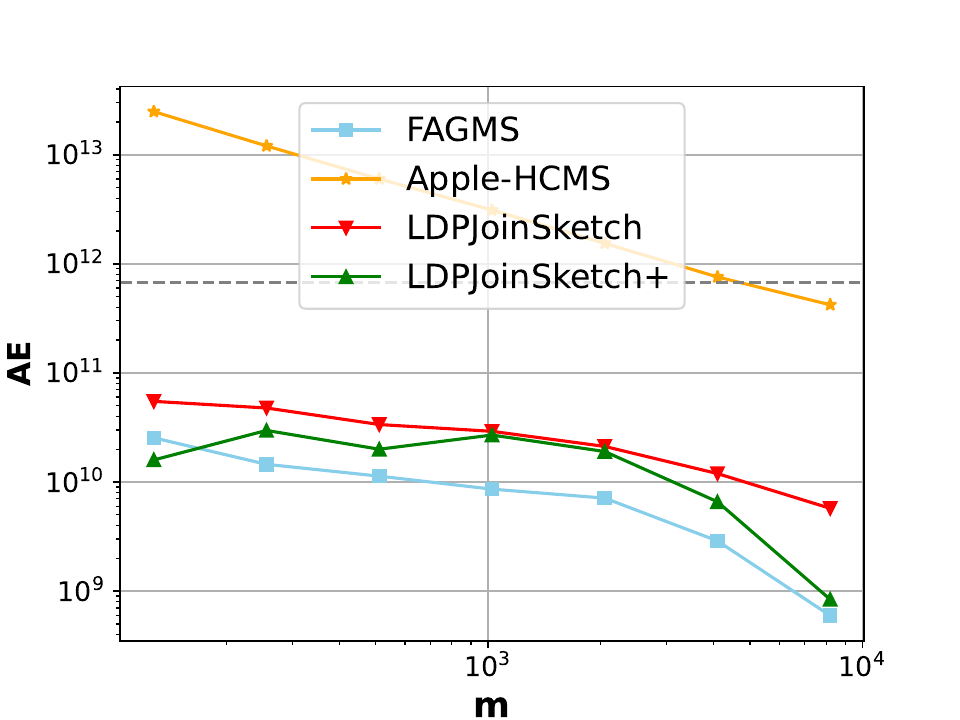}
      \end{minipage}
      \label{fig:m_uniform}
    }
    \subfigure[Twitter]{
        \centering
        \begin{minipage}[b]{0.22\textwidth}
      \includegraphics[width=1\textwidth, trim={0 0 1.2cm 1cm}, clip]{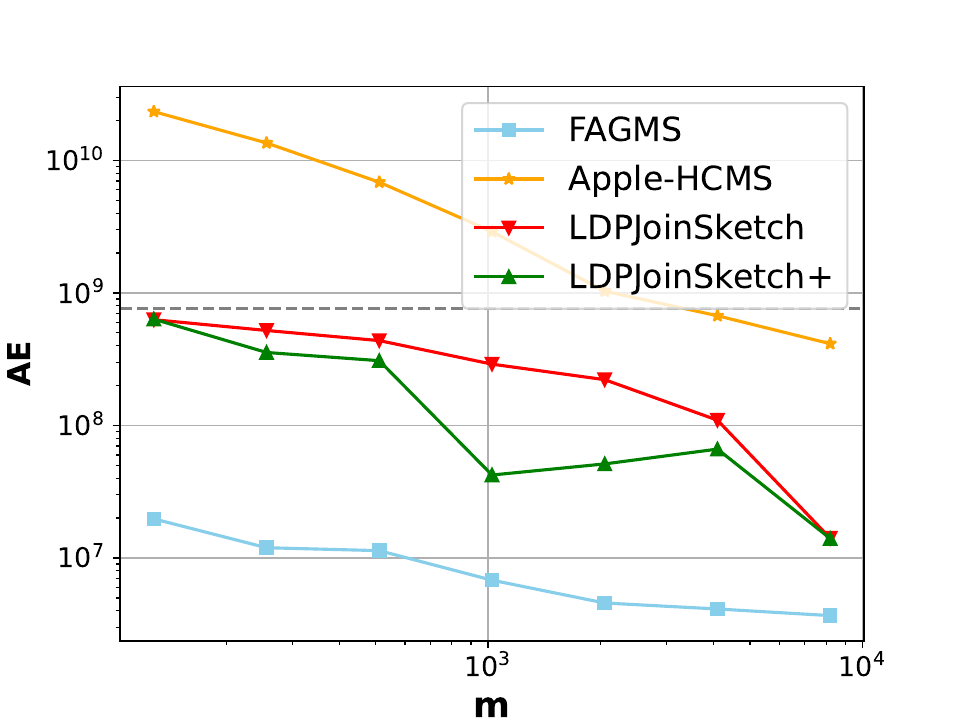}
      \end{minipage}
      \label{fig:m_real}
    }
  \subfigure[Zipf($\alpha$=1.1)]{
      \centering
      \begin{minipage}[b]{0.22\textwidth}
    \includegraphics[width=1\textwidth, trim={0 0 1.5cm 1cm}, clip]{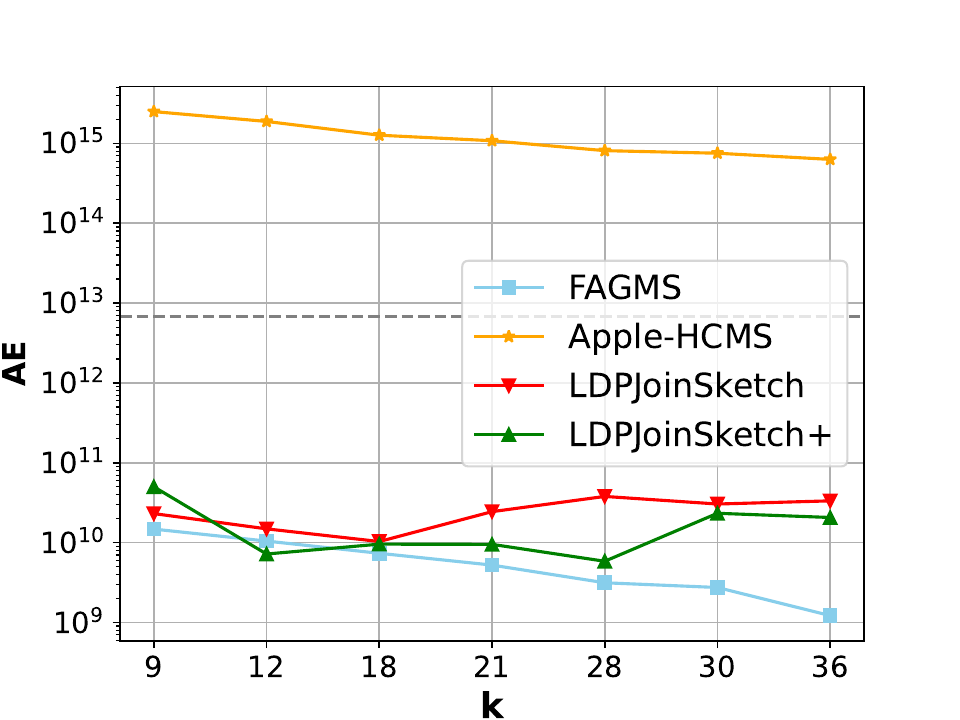}
    \end{minipage}
    \label{fig:k_zipf1_1}
  }
  \subfigure[Zipf($\alpha$=2.0)]{
      \centering
      \begin{minipage}[b]{0.22\textwidth}
    \includegraphics[width=1\textwidth, trim={0 0 1.5cm 1cm}, clip]{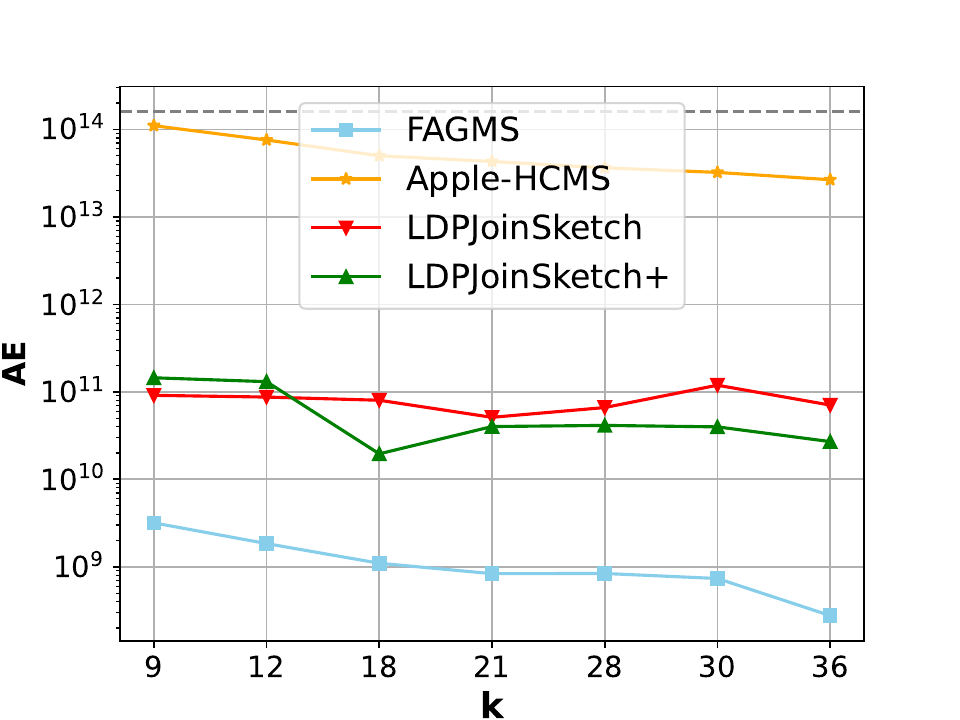}
    \end{minipage}
    \label{fig:k_zipf2_0}
  }
  \subfigure[MovieLens]{
      \centering
      \begin{minipage}[b]{0.22\textwidth}
    \includegraphics[width=1\textwidth, trim={0 0 1.5cm 1cm}, clip]{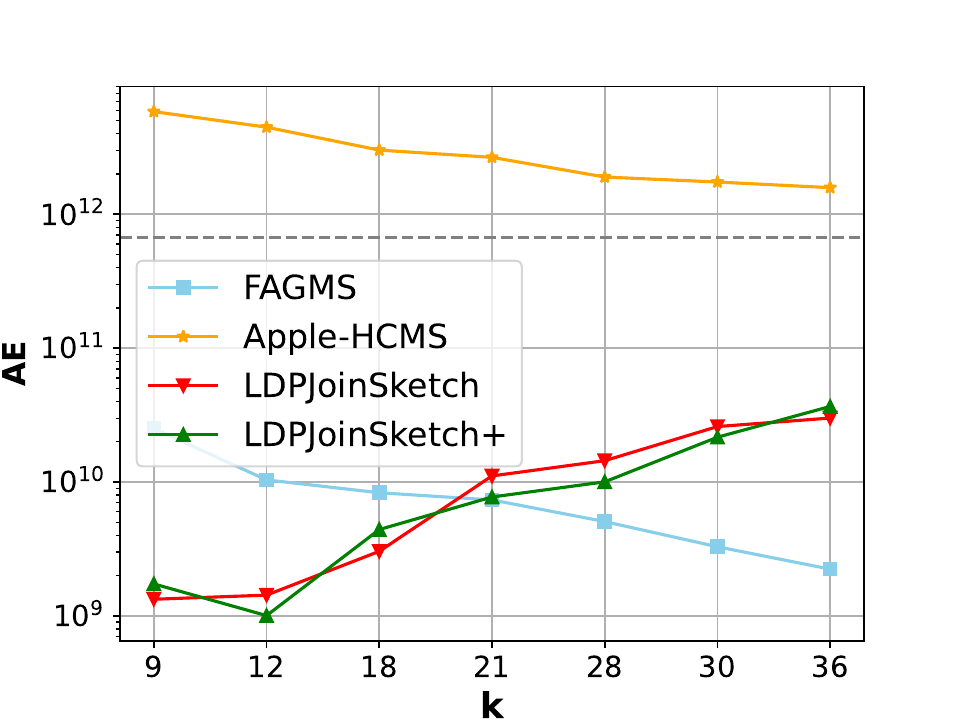}
    \end{minipage}
    \label{fig:k_movielens}
  }
  \subfigure[Twitter]{
      \centering
      \begin{minipage}[b]{0.22\textwidth}
    \includegraphics[width=1\textwidth, trim={0 0 1.5cm 1cm}, clip]{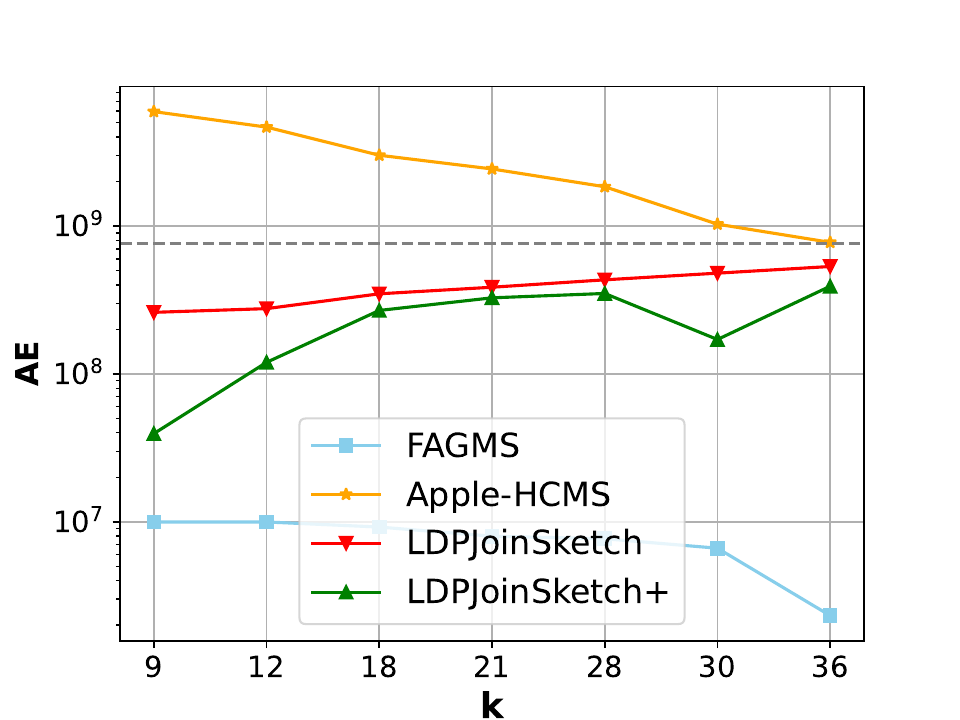}
    \end{minipage}
    \label{fig:k_twitter}
  }
  \caption{Impact of $(m,k)$. (Figures (a)-(d) show impact of $m$, figures (e)-(h) show impact of $k$.)}
\label{fig:k_and_m_impact}
\end{figure*}

\subsection{Impact of parameters}
\noindent \textbf{The impact of privacy budget $\epsilon$}. We report the accuracy of different methods with privacy budget $\epsilon$ varying in $\{0.1,1,...,10\}$ in Fig.~\ref{fig:epsilon_impact}. We set ($k=18$, $m=1024$). Overall, the accuracy improves as $\epsilon$ increases because less noises are introduced.  Our methods perform better when $\epsilon$ is relatively small. 
And LDPJoinSketch+ shows a notable improvement in accuracy on skewed data. We can also find that the error of sketch-based algorithms do not varies much when $\epsilon$ is large.

\noindent \textbf{The impact of sketch parameters $(k,m)$}. Fig.~\ref{fig:k_and_m_impact} show the impact of sketch size on estimation results. First, Fig.~\ref{fig:k_and_m_impact}.(a)-(d) show the impact of hash domain size $m$ on accuracy with a fixed $k=18$. We set $\epsilon=10$ and $r=0.1$. We find that LDPJoinSketch+ has optimal utility on both synthetic and real-world datasets under different $m$.  The error of all methods decreases as $m$ increases, because of less hash collisions.
Second, Fig.~\ref{fig:k_and_m_impact}.(e)-(h)  illustrate the impact of the number of hash functions $k$ with a fixed $m=1024$. As mentioned in theorem~\ref{theorem:error_bound}, the parameter $k$ in our methods indicates the failure probability $\delta$, thus we set $\delta \in \{0.1,0.05,0.01, ..., 0.0001\}$ and correspondingly the number of hash functions is $k \in \{9,12,18,21,28,30,36\}$. Our methods performs best in most cases. For both FAGMS and Apple-HCMS, the errors decrease with a larger $k$. But for our methods, the error does not change much or just slightly increases along with the increase of $k$. The reason is that we only update one counter of the $j$th line in the sketch for each entry, and $j$ is randomly chosen in $[k]$. As $k$ increases, the error caused by sampling increases.

\begin{figure}
  \centering
  \begin{minipage}{0.24\textwidth}
    \centering
  \includegraphics[width=\linewidth, trim={0 0 0 0cm}, clip]{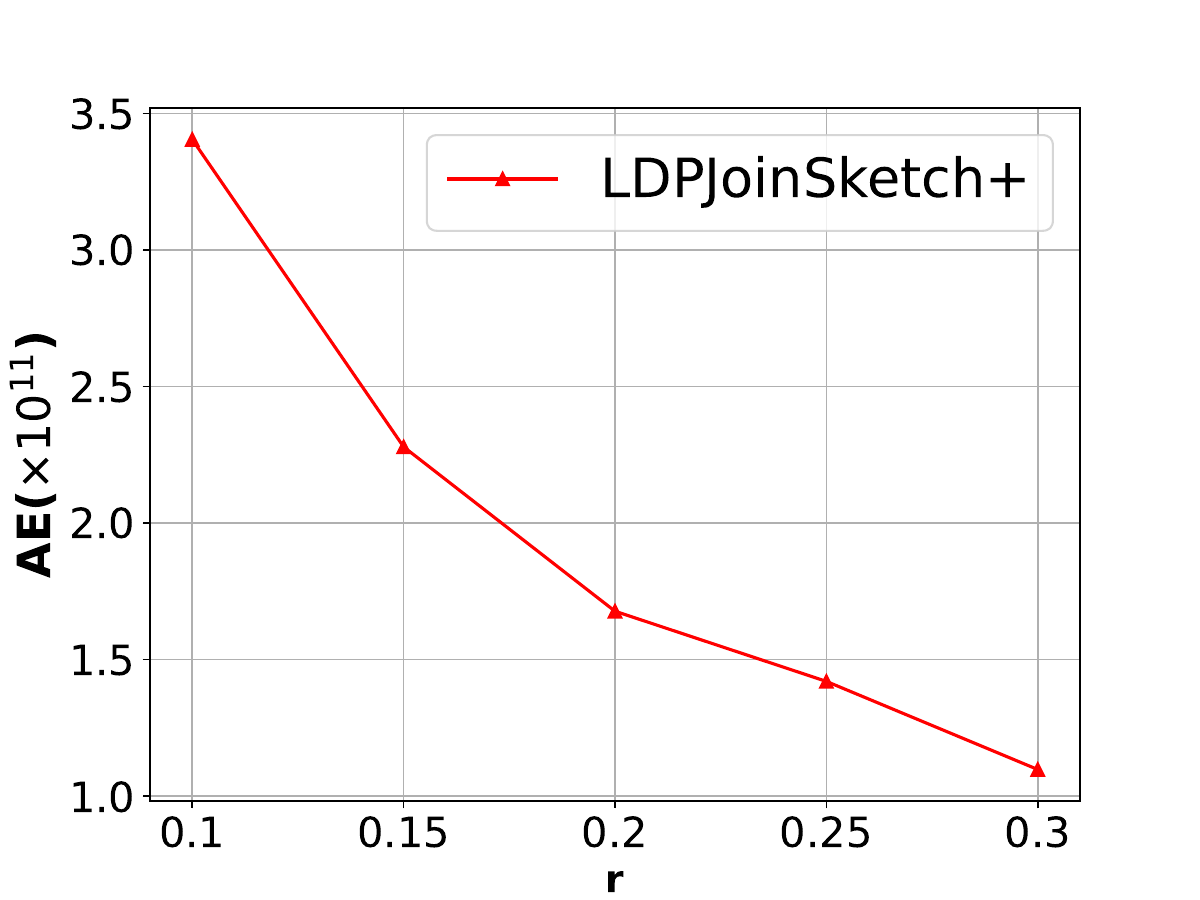}
  \caption{Impact of $r$.}
  \label{fig:sample_rate_impact}
  \end{minipage}
  \hfill
  \begin{minipage}{0.24\textwidth}
  \centering
  \includegraphics[width=\linewidth, trim={0 0 0 1.25cm}, clip]{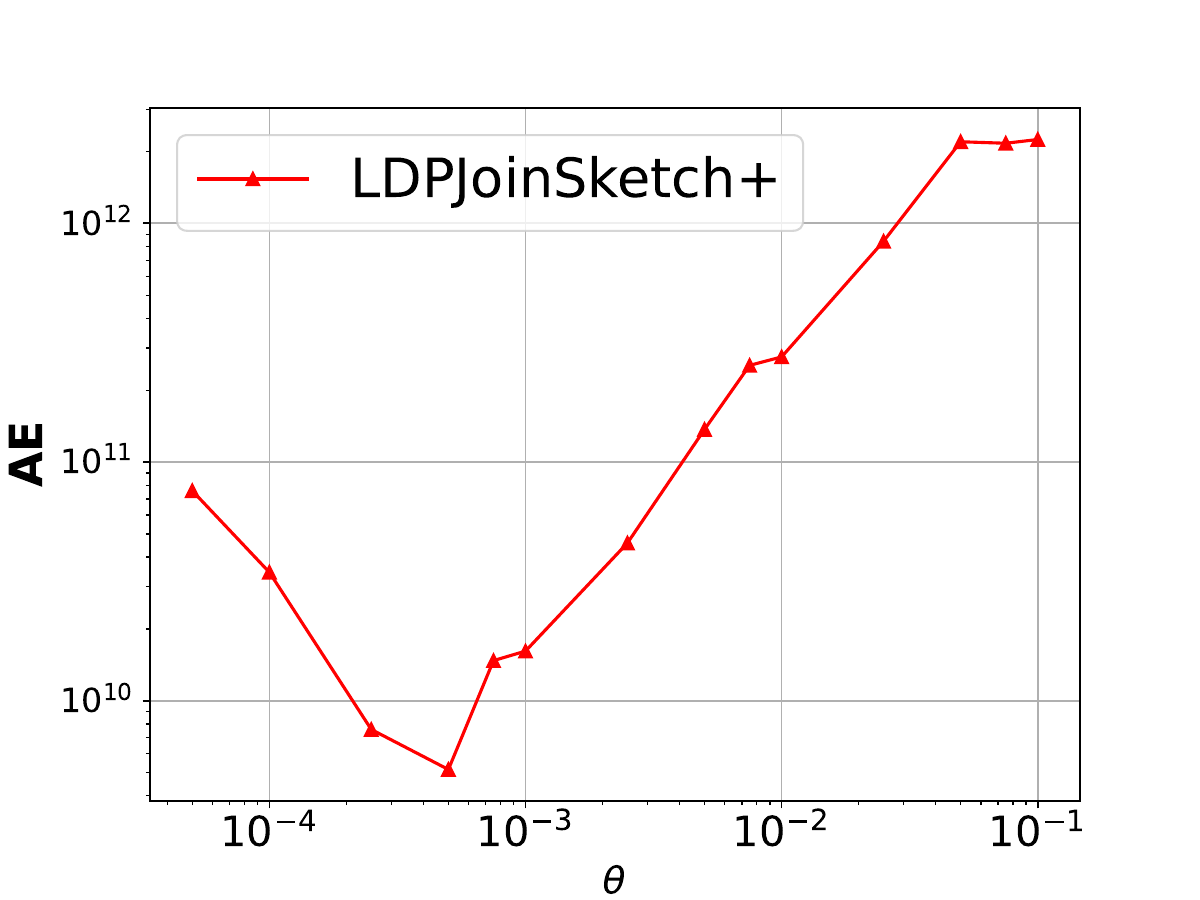}
  \caption{Impact of $\theta$.}
  \label{Fig:threshold}
  \end{minipage}
\end{figure}

\noindent \textbf{The impact of sampling rate $r$}. Fig.~\ref{fig:sample_rate_impact} illustrates the impact of sampling rate $r$ of LDPJoinSketch+ on the accuracy. As the sampling rate in phase 1 affects the frequency estimation accuracy, it consequently influences the FAP in phase 2. We set ($k=18$, $m=1024$), $\epsilon=4$, and vary $r\in \{0.1,0.15,...,0.30\}$ to test the accuracy with Zipf($\alpha$=1.1) dataset. We can observe that the accuracy increases with larger sampling  rate. 

\noindent \textbf{The impact of threshold $\theta$}.  Fig.~\ref{Fig:threshold} shows the results of LDPJoinSketch+ on Zipf($\alpha$=1.1) dataset with threshold $\theta$ ranging from $5\times 10^{-5}$ to 0.1. LDPJoinSketch+ uses $\theta$ to separate the high-frequency and low-frequency items. Here ($k=18$, $m=1024$) and $\epsilon=4$.
On one hand, a smaller $\theta$ causes some low-frequency items will be considered as high-frequency ones, thereby reducing the accuracy of LDPJoinSketch+. On the other hand, a larger  $\theta$ results in fewer items in frequent item set, thus making it less effective in mitigating hash-collision errors. Therefore, it is essential to select an appropriate threshold $\theta$ tailored to the data distribution.

\begin{figure}[hptb]
  \centering
  \begin{minipage}{0.24\textwidth}
    \centering
  \includegraphics[width=\linewidth, trim={0 0 1cm 1cm}, clip]{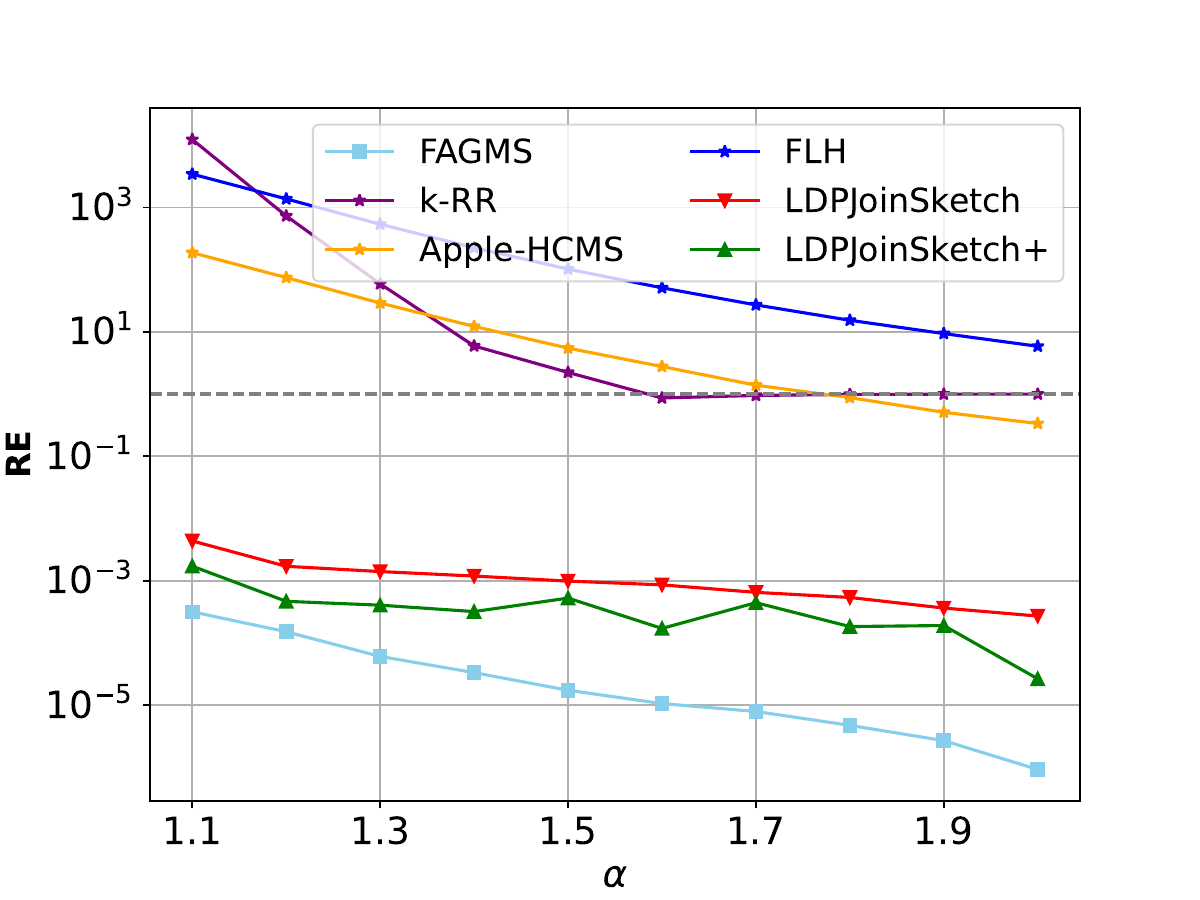}
  \caption{Impact of skewness.}
\label{fig:skewness_impact}
  \end{minipage}
  \hfill
  \begin{minipage}{0.24\textwidth}
    \centering
  \includegraphics[width=\linewidth, trim={0 0 1cm 1cm}, clip]{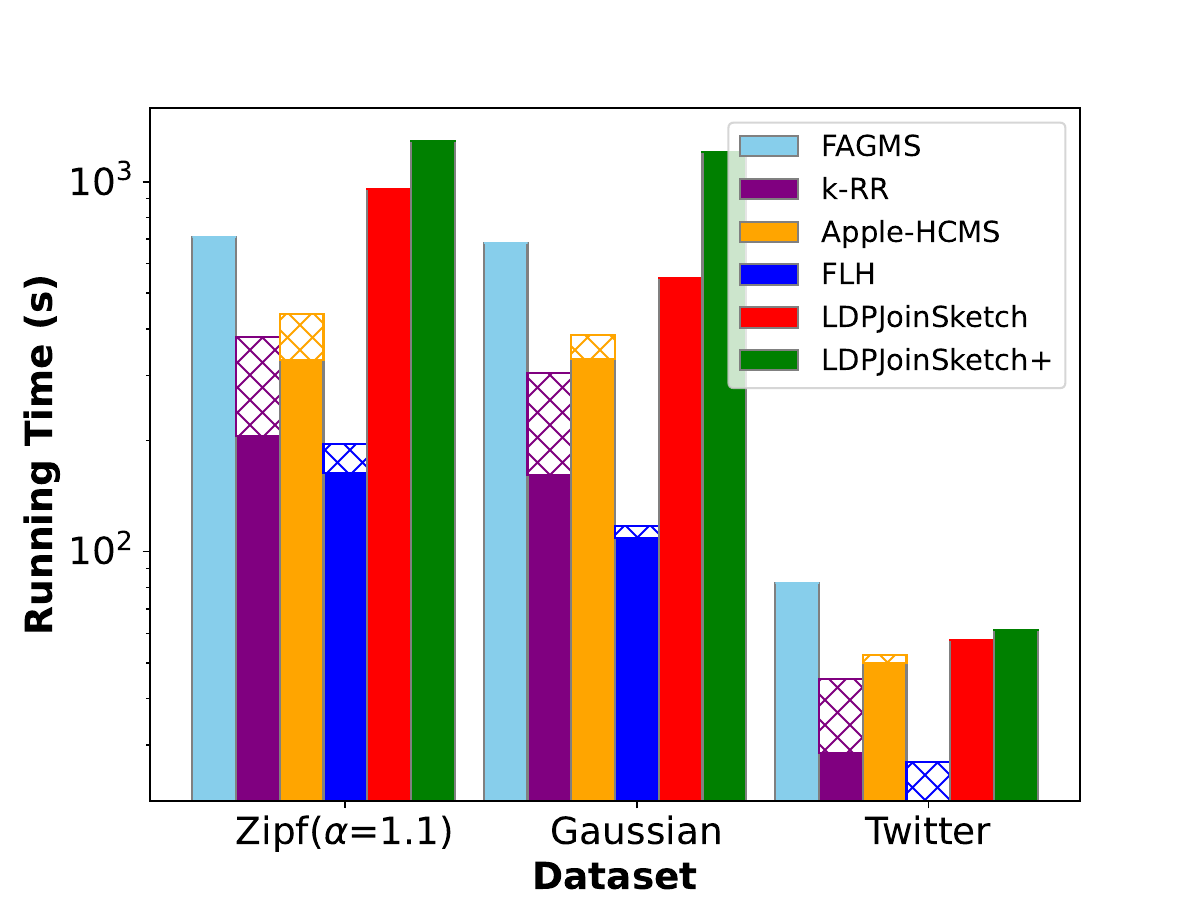}
  \caption{Efficiency}
  \label{Fig:Efficiency}
  \end{minipage}
\end{figure}

\noindent \textbf{The impact of skewness $\alpha$}. Fig.~\ref{fig:skewness_impact} shows the result on Zipf datasets with different level of skewness. We set ($k=18$, $m=1024$) and $\epsilon=4$. 
Our methods achieves the optimal performance with different skewness levels. As the skewness increases, the error of all methods decreases because the true join size increases considerably with larger $\alpha$, meanwhile, the number of distinct items decreases resulting in smaller errors. 




\subsection{Efficiency}
Fig.~\ref{Fig:Efficiency} shows the efficiency of different methods on three datasets. The solid areas in the figure represent off-line running time including collecting perturbed values and constructing sketches, while the grid areas represent online running time of join size estimation. From the figure, we can observe that the online time cost for all sketches-based method is nearly negligible, indicating that the response can be generated quickly after the sketches constructed, which is valuable in interactive settings. Compared with other LDP mechanisms, our algorithms require a little more off-line time, but considering the huge improvement in the accuracy of estimation and online querying, the extra time is well spent.


\begin{figure}[htbp]
  \centering
  \subfigure[Zipf($\alpha$=1.5)]{
      \centering
      \begin{minipage}[b]{0.22\textwidth}
    \includegraphics[width=1\textwidth, trim={0 0 1cm 1cm}, clip]{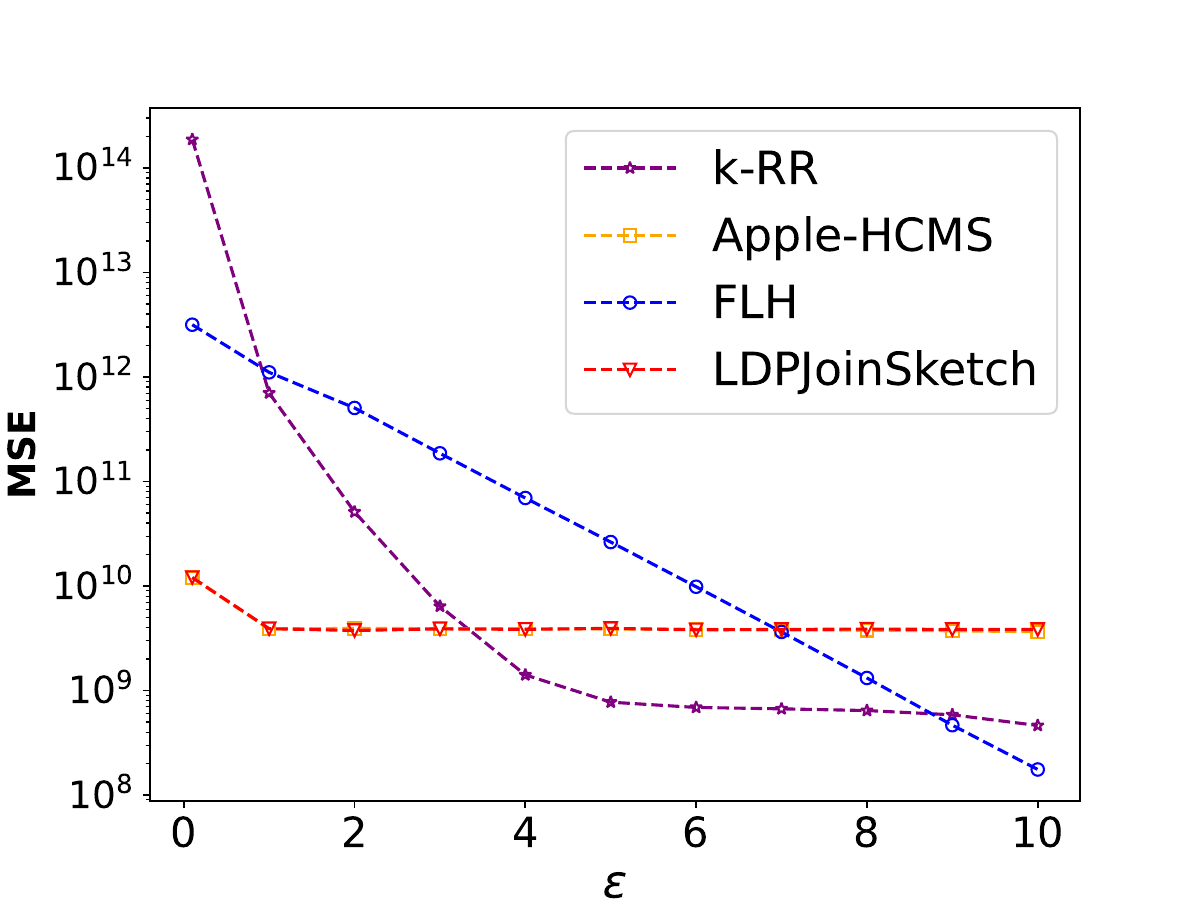}
    \end{minipage}
    \label{fig:mse_zipf}
  }
  \subfigure[MovieLens]{
      \centering
      \begin{minipage}[b]{0.22\textwidth}
    \includegraphics[width=1\textwidth, trim={0 0 1cm 1cm}, clip]{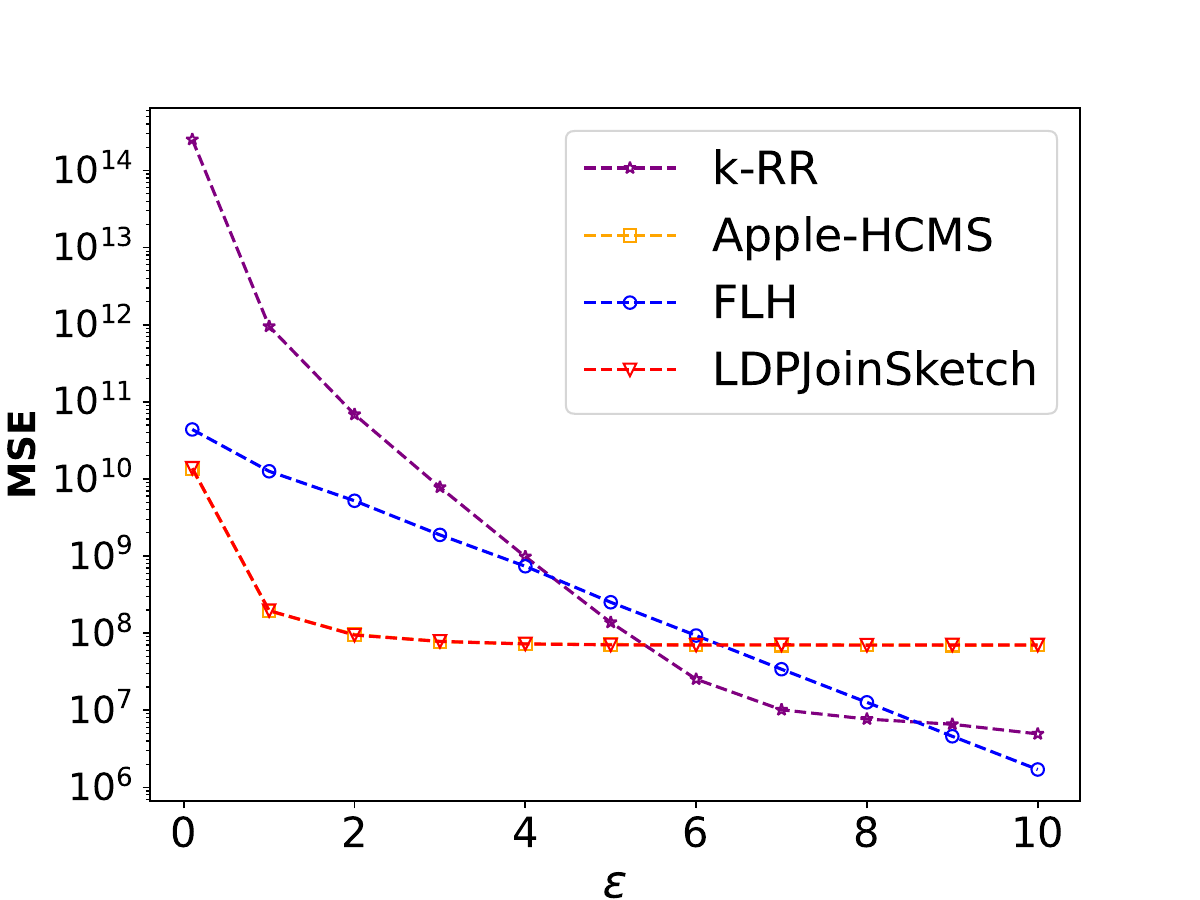}
    \end{minipage}
    \label{fig:mse_movielens}
  }
  \caption{Frequency Estimation.}
\label{Fig:frequency_estimation}
\end{figure}
\subsection{Frequency Estimation}
Fig.~\ref{Fig:frequency_estimation} shows the performance of LDPJoinSketch on frequency estimation compared with other LDP mechanisms. We can learn that LDPJoinSketch has the same accuracy level as Apple-HCMS since both two algorithms have almost identical sketch structures except the additional hash function $\xi$ in LDPJoinSketch. Moreover, LDPJoinSketch is more accurate when $\epsilon$ is small. The error of LDPJoinSketch and Apple-HCMS do not change much when $\epsilon$ reaches a certain value, because the error is mainly introduced by the sketch rather than the privacy budget given a large $\epsilon$. 

\begin{figure}[htpb]
  \centering
  \includegraphics[scale=0.35]{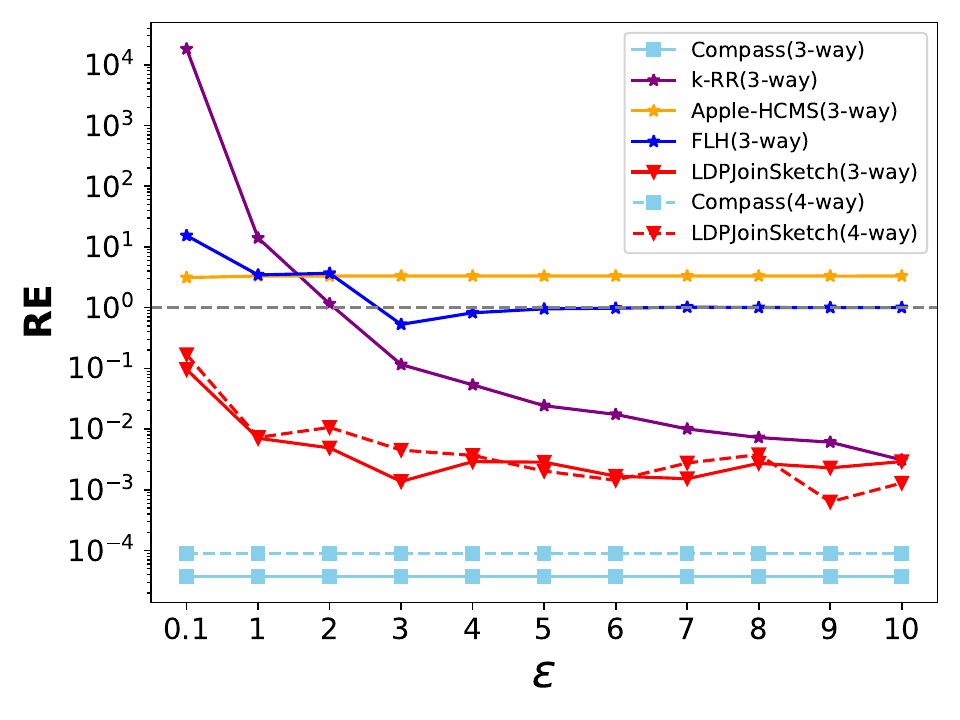}
  \caption{varying $\epsilon$ on Multi-Way Join.}
  \label{Fig:Fig:multi-way eps}
\end{figure}
\subsection{Experiments on Multi-way join}
Fig.~\ref{Fig:Fig:multi-way eps} illustrates the performance of various methods for multi-way chain join on Zipf($\alpha=1.5$) dataset. Here we use 3-way and 4-way chain join queries for evaluation. Due to the extremely high time cost of frequency-based methods such as k-RR, we simply show the performance of Compass and LDPJoinSketch in 4-way join cases. An example of a 3-way join is like $T_1(A)\Join T_2(A,B)\Join T_3(B)$.
 We set ($k=18$, $m=1024$). The dashed lines in the figure represent the 4-way join cases. We can observe that our LDPJoinSketch is still effective and performs the best. Besides, the estimation error decreases continuously as $\epsilon$ increases and eventually becomes stable, because the noise at this point mainly comes from the sampling operation of sketch.


\subsection{Summaries of experimental results.}
\begin{itemize}
  \item [$\bullet$] LDPJoinSketch is more accurate than k-RR, FLH, and Apple-HCMS for join size estimation on sensitive data. 

  \item [$\bullet$] Our improved method LDPJoinSketch+ further decreases the error introduced by hash collisions through frequency-aware perturbation mechanism.

  \item [$\bullet$] Both LDPJoinSketch and LDPJoinSketch+ we proposed are better suited for large datasets whose join values are sensitive and have large domains. 
\end{itemize}

\section{Conclusion}\label{sec:Conclusion}
This paper proposes two sketch-based methods, namely LDPJoinSketch and LDPJoinSketch+, for join size estimation under LDP. LDPJoinSketch is tailored to handle private join values with large domains, while LDPJoinSketch+ enhances accuracy by mitigating hash collisions within the LDP framework. Join size estimation, exemplified by the aggregation function ``COUNT'', serves as an initial focal point. Looking ahead, our future research will focus on addressing generate join aggregation queries under local differential privacy, demanding additional diligence and investigation.

\ifCLASSOPTIONcompsoc
  \section*{Acknowledgments}
\else
  \section*{Acknowledgment}
\fi

This research was partially supported by the NSFC grant 62202113; GuangDong Basic and Applied Basic Research Foundation SL2022A04J01306; Open Project of Jiangsu Province Big Data Intelligent Engineering Laboratory SDGC2229; the Guangzhou Science and Technology Plan Project (No. 2023A03J0119); National Key Research and Development Program of China (No. 2022YFB3104100).

\bibliographystyle{IEEEtran}
\bibliography{mybibfile}

\end{document}